\documentclass[10pt, journal, doublecolumn]{IEEEtran}

\ifCLASSINFOpdf
  \usepackage[pdftex]{graphicx}
 \graphicspath{{../figures/}{../jpeg/}}
\else
\fi

\usepackage[cmex10]{amsmath}
\usepackage{amsfonts}
\usepackage{amsthm}
\usepackage{cite}
\usepackage{cases}
\usepackage{caption}
\usepackage{dsfont}
\usepackage{mathrsfs}
\usepackage{subfigure}

\usepackage{color}

\DeclareMathOperator{\arccot}{arccot}

\theoremstyle{proposition}
\newtheorem{remark}{Remark}
\theoremstyle{proposition}

\theoremstyle{proposition}
\newtheorem{theorem}{Theorem}
\theoremstyle{proposition}
\newtheorem{lemma}{Lemma}

\begin{document}

\title{Information Exchange in Randomly Deployed Dense WSNs with Wireless Energy Harvesting Capabilities}
%
%
%

\author{Prodromos-Vasileios~Mekikis, \IEEEmembership{ Student Member, IEEE},
Angelos Antonopoulos, 
Elli Kartsakli, \IEEEmembership{ Senior Member, IEEE},
Aris S. Lalos, \IEEEmembership{ Member, IEEE},
Luis Alonso,
Christos Verikoukis, \IEEEmembership{ Senior Member, IEEE}
\thanks{
This work was supported by the projects Cellfive (TEC2014-60130-P), AGAUR (2014-SGR-1551) and ESEE(324284).

P.-V. Mekikis, E. Kartsakli,  A. S. Lalos and L. Alonso are with the Department of Signal Theory and Communications (TSC), Technical
University of Catalonia (UPC), Spain, (e-mail: \{vmekikis, ellik, aristeidis.lalos, luisg\}@tsc.upc.edu).

A. Antonopoulos and C. Verikoukis are with the Telecommunications Technological Centre of Catalonia (CTTC), Spain, (e-mail: \{aantonopoulos, cveri\}@cttc.es). }}
%
%

\markboth{}%
{Shell \MakeLowercase{\textit{et al.}}: Bare Demo of IEEEtran.cls for Journals}
%



\maketitle
\thispagestyle{empty}
\IEEEdisplaynotcompsoctitleabstractindextext
\begin{abstract}
As large-scale dense and often randomly deployed wireless sensor networks (WSNs) become widespread, local information exchange between co-located sets of nodes may play a significant role in handling the excessive traffic volume. Moreover, to account for the limited life-span of the wireless devices, harvesting the energy of the network transmissions provides significant benefits to the lifetime of such networks. In this paper, we study the performance of communication in dense networks with wireless energy harvesting (WEH)-enabled sensor nodes. In particular, we examine two different communication scenarios (direct and cooperative) for data exchange and we provide theoretical expressions for the probability of successful communication. Then, considering the importance of lifetime in WSNs, we employ state-of-the-art WEH techniques and realistic energy converters, quantifying the potential energy gains that can be achieved in the network. Our analytical derivations, which are validated by extensive Monte-Carlo simulations, highlight the importance of WEH in dense networks and identify the trade-offs between the direct and cooperative communication scenarios.
\end{abstract}

\begin{IEEEkeywords}
Wireless Sensor Networks, Cooperative Networks, Wireless Energy Harvesting, Dynamic Power Splitting, Realistic RF-to-DC Conversion Efficiency, Stochastic Geometry
\end{IEEEkeywords}

%
\IEEEpeerreviewmaketitle

\section{Introduction}\label{intro}

\IEEEPARstart{D}{espite} their limited processing and energy capabilities, wireless sensor networks (WSNs) apply in an increasing number of domains, such as environmental monitoring \cite{critical}, mobile healthcare \cite{health} and intelligent transportation systems \cite{its}. With the introduction of new paradigms, such as Machine-to-Machine communication and Internet of Things, the number of wireless nodes in WSNs increases constantly, creating large-scale and dense randomly deployed networks. In such networks, the interference and the excessive traffic can significantly affect the quality of service (QoS) and, consequently, the network lifetime. Therefore, although in typical WSN scenarios the information collected by the sensors is forwarded through the network to a central control station (sink) for centralized handling and decision-making, recent applications in dense networks drive the need for local data exchange among the nodes. To that end, many works have been motivated to consider the use of distributed algorithms that encourage data processing on the node side. Distributed estimation \cite{disest}, distributed clustering \cite{clust} and distributed data storage \cite{dds} are among the applications that support local data exchange and processing to improve the network performance and the energy efficiency. Hence, the design of effective schemes that enable neighboring nodes to exchange messages and apply distributed algorithms locally is becoming of considerable importance \cite{tvtbidi}.

Given the dense deployment, it is very probable that surrounding nodes overhear the transmissions of the network and are willing to assist the communication by acting as relays. This concept, known as cooperation \cite{sendo}, can provide noteworthy gains in the communication and was initially studied in small-scale networks where the relays are deployed in favorable positions (e.g., in between the transmitting nodes) \cite{qosrelay}. However, in large-scale networks, i) relay selection needs considerable overhead and signaling \cite{relsel} and ii) it is hard to maintain a favorable position of the relays for every pair of the randomly deployed transmitting nodes. Nevertheless, although cooperation cannot always guarantee notable performance gains in large-scale dense networks \cite{notcoop}, it is possible to achieve diversity gains that increase the network reliability.

Besides, due to the limited human intervention for practical matters (e.g., replacing batteries), energy efficient communication becomes an essential concern in the design of large-scale networks. Although cooperation can improve the energy efficiency of a WSN, there are more effective ways to extend the network lifetime, which is a key parameter of a WSN and strongly depends on the limited-capacity batteries. Currently, a popular and drastic way to prolong the network operation is by harvesting energy from the environment to either power entirely the sensor nodes or extend the lifetime of the existing batteries \cite{wsneh1,wsneh2,wsneh3}. In this new paradigm, which is broadly known as energy harvesting (EH), the most typical energy sources are solar, thermal, wind and kinetic energy. Recently, wireless energy harvesting (WEH) \cite{weh} has emerged as an alternative approach to harvest the energy of the electromagnetic radiation (EMR) from the network transmissions without the need of expensive hardware systems. WEH can be adopted even in cases where the aforementioned energy sources are scarce or unstable due to their dependence on stochastic events like the weather conditions. This constitutes it a reasonable and straightforward method to extend the lifetime of the wireless nodes and, consequently, of the whole network. 
 
Due to the dependence of the energy conversion efficiency of the harvester on the amount of received EMR \cite{rectenna,rectenna2}, the benefits from WEH are marginal for small-scale network applications, but interestingly high for large-scale dense networks. Ideally, with WEH, it would be possible to improve vastly the network performance by simultaneously transferring information and harvesting all the power. However, since the reuse of the whole received signal both from the energy harvester and the information receiver is not yet possible, various methods have been proposed in order to facilitate WEH \cite{editorgen}. In the class of these techniques, dynamic power splitting (DPS) \cite{editor} has been proved to be among the most efficient approaches that facilitates simultaneous message decoding and energy harvesting. Using DPS, it is possible to dynamically share the received energy between the information decoder and the energy harvester, according to the channel condition that is assumed to be known at the receiver.

To that end, several studies that consider large-scale networks with WEH have lately appeared in the literature \cite{khemr,emerg,krik,mekikis}. In his pioneer work \cite{khemr}, Huang studies the network throughput in a basic mobile ad hoc scenario, where the communication between the transmitter and the receiver is conducted through an ideal wireless channel (i.e., no path loss is assumed in the link). It is worth noting that, although some of the potential benefits of the WEH technology are identified in \cite{khemr}, the results cannot be generalized for cooperative communications. Particularly, in cooperative scenarios, the existence of relay nodes imply a volatile and complex environment that requires a dedicated study. Similarly, in \cite{emerg}, Guo and Wang study the effects of WEH in a direct communication scenario. Nevertheless, the analysis is based on specific physical layer configurations, since the authors provide closed-form expressions for the QoS metrics only for specific path loss conditions, i.e., a particular value for the path loss exponent. However, the range of values that the path loss exponent can have in different environments stresses the need for theoretical expressions that provide general and environment-independent solutions. Recently, an interesting approach has been presented in \cite{krik} by Krikidis, where the coverage of a large-scale network is studied, while the receivers employ a technique for simultaneous information and energy transfer. The author provides incentives for cooperation, highlighting the possible benefits, however the proposed model considers fixed distances between preassigned nodes. In addition, the model assumes a constant energy conversion efficiency for the harvester, although in realistic implementations the efficiency depends on the input power. In the same context, the work in \cite{mekikis} studies a bidirectional scenario with nodes that harvest EMR with a constant energy conversion efficiency. The authors provide important insights into the probability of data exchange in such scenarios, but there is no analysis with regard to the end-to-end network performance, which is essential for the evaluation of the proposed model. In addition, the possibility of direct communication among the randomly deployed nodes is neglected, as only cooperative operation is considered.

In this paper, we study the impact of WEH using DPS on the information exchange in large-scale networks. We consider two sets of sources that exchange their messages either directly or via randomly deployed relay nodes. As performance gains from cooperation are not always guaranteed in dense networks, it is interesting to investigate the potential benefits of cooperation in a WEH-enabled dense network. In addition, we employ a realistic model for the WEH conversion efficiency of the receivers \cite{rectenna2}. Our contribution can be summarized in the following points:

\begin{itemize}
  \item We analytically derive the probability of successful data exchange, while taking into account DPS.
  \item In order to demonstrate the potential energy gains of WEH, we analytically estimate the network lifetime with and without WEH. We assume a variable and, thus, realistic energy conversion efficiency for the harvester to comply with state-of-the-art rectennas.
   \item We provide theoretical expressions for a well-established end-to-end QoS performance metric, namely the spatial throughput, and derive theoretically the optimal intensity that maximizes the network lifetime.
  \item We conduct an extensive performance assessment for the two schemes (direct and cooperative), which reveals intriguing trade-offs that provide useful insights for the design of WSNs with WEH.
\end{itemize}

The rest of the paper is organized as follows. Section \ref{section_systemmodel} describes the system model and the communication scenarios. Section \ref{section_pex} presents the analysis for the probability of successful message exchange. Section \ref{section_lifetime} includes the theoretical expressions of the average network lifetime for the different scenarios, while, in Section \ref{section_metrics}, we present useful performance metrics. Section \ref{section_results} presents the model validation and the numerical results. Finally, Section \ref{section_conclusion} concludes the paper.
\section{System Model}\label{section_systemmodel}
\subsection{Network and Channel Model}
We consider a large-scale network consisting of two sets of source nodes $S_1=\{s_{11},\dots,s_{1i}\}$, $S_2=\{s_{21},\dots,s_{2j}\}$ and a set of ambient nodes acting as relays $R=\{r_{1},\dots,r_{k}\}$ in two different communication scenarios: i) direct, where the sets of source nodes exchange messages directly, and ii) cooperative, where the randomly deployed relays $R$ assist $S_1$ and $S_2$ to the message exchange. In cases where it is convenient, a set of sources will be denoted as $S_\varphi$, $\varphi\in\{1,2\}$ while $S_{\hat{\varphi}}$ will denote the complementary set (i.e., when $S_\varphi=S_2$ then $S_{\hat{\varphi}}=S_1$ and vice versa). The relays are assumed to be other sensor nodes or other type of devices (e.g., smartphones with dedicated interface for relaying). The different sets of sources measure different phenomena and broadcast their measurements. More specifically, each individual source node receives a local measurement, either directly or cooperatively from the nearest node of the other type (i.e., nearest-neighbor model \cite{nearestmodel}). Consequently, each node is required to be aware of the location of itself and of its neighbors, via localization schemes that act in higher network layers\cite{loc}.

All nodes are identical and assumed to be moving on the same Euclidean plane. They are represented by three independent homogeneous PPPs, a reasonable approach for WSNs according to \cite{sgtutorial}. The locations of the sources $S_1$ are described by the PPP $\Phi_{S_1}=\{x_1,\dots , x_i\}$ with intensity $\lambda_1$, where $x_i$, $\forall i \in \mathbb{N}$, denotes the location of the source $s_{1i}$ on the plane $\mathbb{R}^2$. Similarly, the location of the sources $S_2$ on $\mathbb{R}^2$ are represented by the PPP $\Phi_{S_2}=\{y_1,\dots , y_j\}$ with intensity $\lambda_2$, where $y_j$, $\forall j \in \mathbb{N}$ denotes the location of the source $s_{2j}$. For the modeling of the relay nodes, there is an additional PPP $\Phi_R=\{z_1,\dots , z_k\}$ with intensity $\lambda_R$, which represents the location $z_k$, $\forall k \in \mathbb{N}$, of the relay $r_k$. 

For our analysis, without loss of generality, we assume that the respective receiving node in each slot is located at the origin (Slyvnyak's theorem \cite{book}). The received power $P_{R}$ at a node located in a distance $d$ from the transmitting node is $P_{R}=P_thd^{-\alpha}$, where $P_t$ is the transmission power of the nodes, $\alpha>2$ is the path loss exponent and $h$ is the square of the amplitude fading coefficient (i.e., the power fading coefficient) that is associated with the channel between the nodes. We also assume that the fading coefficients are independent and identically distributed (i.i.d.). Moreover, the amplitude fading $\sqrt{h}$ is Rayleigh with a scale parameter $\sigma=1$, hence $h$ is exponentially distributed with mean value $\mu$. The channel is assumed to remain constant in one time slot (i.e., a time period in which a transmission takes place).

\begin{figure}[!t]
\includegraphics[width=0.8\columnwidth]{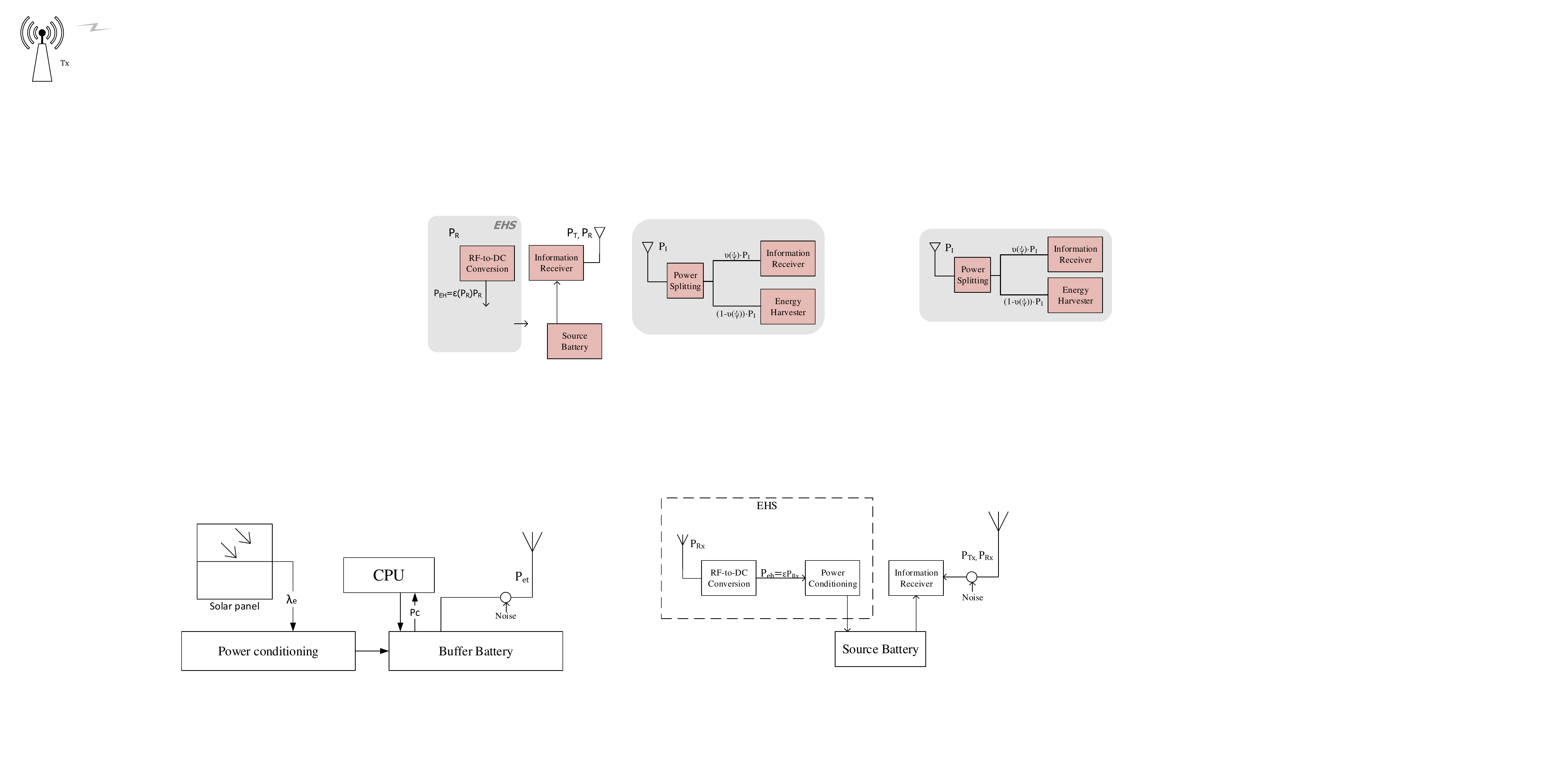}
\centering
\caption{Schematic of a node at reception mode. The received power is dynamically split based on the rule given in \eqref{powersplitfunction}.}
\label{ehs}
\end{figure}

All nodes are powered by a battery with initial energy level $B_I$ and in every time slot consume energy to communicate (i.e., $P_t$ power is consumed for transmission and $P_r$ for reception). Also, they are capable of WEH using a power splitter that dynamically adjusts the power ratio that is allocated to the information receiver and the energy harvester, i.e., DPS \cite{editor}. A simplified illustration of a node is provided in Fig. \ref{ehs}, where the various parts of the node are shown. A node is able to recharge its battery by harvesting the EMR energy from the transmissions of the sources and the relays in the network. According to DPS, the splitting depends on the channel condition and it is described by the following rule:
\begin{numcases}{v(\psi)=}\label{powersplitfunction}
    1, &  if $h<\psi$ \nonumber\\
    \frac{\psi}{h}, &  if $h\geq \psi$
 \end{numcases}
 where $h$ is the power fading coefficient of the channel between the receiver and the nearest transmitter and $\psi$ is a parameter that defines the amount of power that is split between the energy harvester and the information receiver. Later in this paper, we provide an empirical method to choose the value of the $\psi$ parameter. In addition, it is assumed that $h$ is known at the receiving node, but unknown to the transmitter. According to \eqref{powersplitfunction}, when the channel conditions are poor, all of the received signal is being fed to the information receiver.  On the contrary, when the channel conditions are satisfactory for the information receiver, then a fraction of the received power equal to $(1-\frac{\psi}{h})\in[0,1]$ is being fed to the energy harvester without deteriorating the communication performance. At this point, we should mention that the employed DPS technique does not necessarily provide optimal performance in terms of harvested energy for our interference-limited system. However, it is a novel technique that considers the impact of fading and, thus, avoids compromising the communication performance.

\begin{figure}[!t]
\includegraphics[width=0.9\columnwidth]{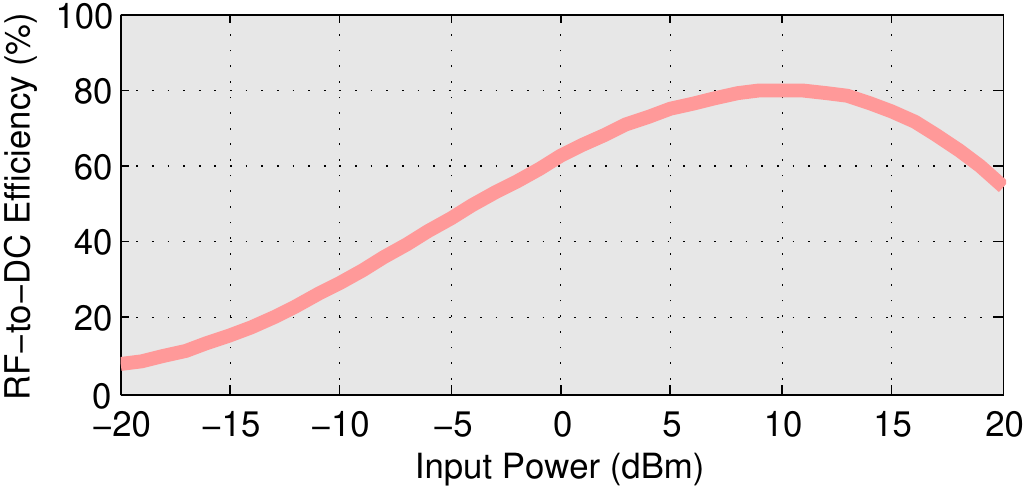}
\centering
\caption{Behavior of the RF to DC efficiency of a rectenna.}
\label{rfdceff}
\end{figure}

Furthermore, the conversion efficiency of the radio frequency (RF) energy into direct current electricity is denoted by $\epsilon$.  As the conversion efficiency of a rectenna depends on the received power \cite{rectenna,rectenna2}, we adopt a variable conversion efficiency $\epsilon$ modeled as a quadratic polynomial that captures the behavior of state-of-the-art rectennas \cite{rectenna,rectenna2}, as in Fig. \ref{rfdceff}, given by
\begin{equation}
\epsilon(P_I)=a_3P_I^3+a_2P_I^2+a_1P_I+a_0,\label{conveff}
\end{equation}
where $P_I$ in Watts is the input power or the total received power, which consists of the received signal and the interference, while $a_3,a_2,a_1,a_0$ are the coefficients of the cubic polynomial. 

After taking into account DPS, a message is considered to be successfully decoded at a receiver when the signal-to-interference-plus-noise ratio (SINR) from its nearest transmitter, denoted as $\gamma$, is higher than a threshold $\gamma^*$; otherwise the message is dropped \cite{physmod1}. The SINR of a mobile node located at the origin at a distance $d$ from its nearest transmitter is
\begin{equation}
\gamma=\frac{v(\psi)P_thd^{-\alpha}}{v(\psi)I_d+N},\label{sinr}
\end{equation}
where $I_d$ is the aggregated interference caused by the transmitter's PPP, defined as $I_d=\sum_{x\in\Phi}P_th_xx^{-\alpha}$ and $N$ is the additive white Gaussian noise power, modeled as a constant zero mean Gaussian Random Variable (RV).

\begin{figure*}[!t]
\includegraphics[width=2\columnwidth]{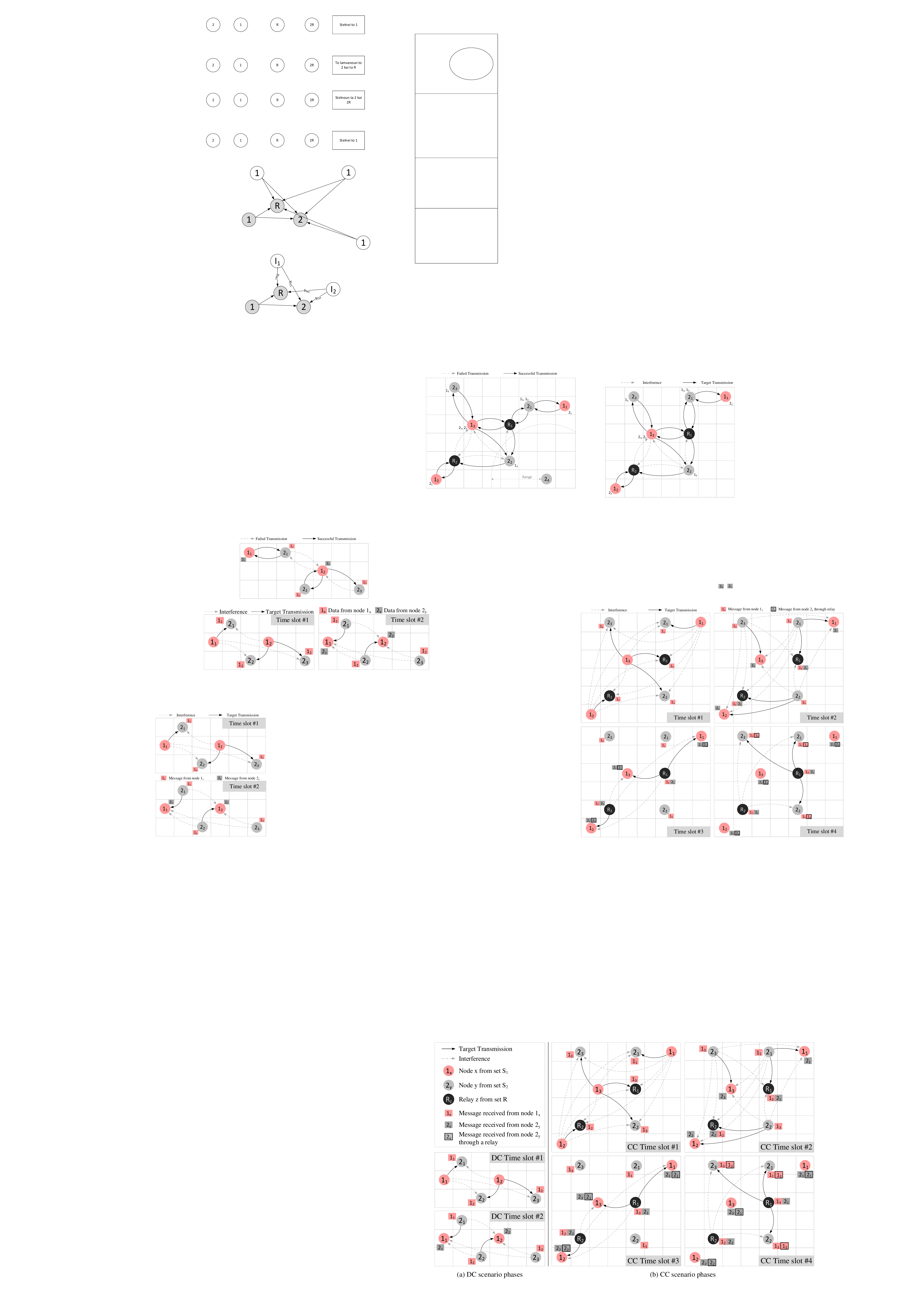}
\centering
\caption{Communication phases. (a) DC scenario phases: i) Slot 1 ($S_1\rightarrow S_2$), ii) Slot 2 ($S_2\rightarrow S_1$), (b) CC scenario phases: i) Slot 1 ($S_1\rightarrow S_2, R$), ii) Slot 2 ($S_2\rightarrow S_1, R$), iii) Slot 3 with active relay ($S_1\leftarrow R$), iv) Slot 4 with active relay ($S_2\leftarrow R$)}
\label{comphases}
\end{figure*}

\subsection{Communication Model}
The time is divided into time slots of fixed duration $t_s$, in which the transmission of one packet can take place. The time needed for the two sets of sources to exchange messages is called communication period (CP). Each CP consists of $g$ time slots, depending on the communication scheme, as we will describe in detail next. 

\subsubsection{Direct communication scenario (DC)}
In the DC scenario, illustrated in Fig. \ref{comphases}a, the CP consists of two time slots (i.e., $g_{DC}=2$) of duration $t_s$. In the first time slot, each $S_1$ source is broadcasting its message and each $S_2$ source attempts to decode the message of its nearest $S_1$ source. The rest transmissions of the $S_1$ sources are considered as interference for the $S_2$ source. However, when the circumstances allow it (i.e., $h\geq\psi$), this interference is beneficial for the network, because a part of it is harvested. In the second time slot, the system follows a similar procedure and each $S_1$ source attempts to decode a message from its nearest $S_2$ source. In the end of the CP, all source nodes have attempted to decode a message from their nearest node of the other type, as it is depicted in Fig. \ref{comphases}a (i.e., small rectangular next to each node). In the second time slot of this figure, it can be noticed that node $2_3$ has attempted to decode the message from its nearest node $1_2$, although the latter has attempted to decode the message of its nearest $S_2$ node, i.e., $2_2$. Therefore, there are not always certain pairs in the network, as it happens with nodes $1_1$ and $2_1$. In this way, all nodes manage to receive a message from their nearest neighbor, which is the goal in such scenarios.

\subsubsection{Cooperative communication scenario (CC)}
In the CC scenario, illustrated in Fig. \ref{comphases}b, the CP consists of four time slots (i.e., $g_{CC}=4$). Similar to the DC scenario, in the first two slots, the $S_1$ and $S_2$ sources are attempting to decode the message from their nearest neighbor of the other type. However, in this scenario, there are also relays distributed on the plane that attempt to decode the messages from their nearest source nodes to assist the communication. Therefore, in the following two time slots, the relays are consecutively broadcasting the messages of their nearest $S_1$ and $S_2$ node. In this way, there is a diversity gain, since the sources have two possible ways of receiving a message from a source of the other type. At the fourth time slot in Fig. \ref{comphases}b, we notice that most source nodes have received the same message twice. This means that these nodes have higher probability to decode this message. However, depending on the random topology, there is a chance that some source nodes will receive two different messages, as it happens in nodes $1_3$ and $2_1$ and, thus, deduce more information about their environment. Moreover, if a relay fails to decode the messages in the first two time slots, then it transmits power to the sources to cooperate only in terms of energy. 

\section{Successful Message Exchange Probability}\label{section_pex}
In this section, we present the probability of successful message exchange between the two types of sources in one CP for the DC and CC scenarios. The successful message exchange probability is an important QoS metric, defined as the probability of both $S_1$ and $S_2$ sources to decode successfully the received messages within a CP. 

\subsection{Direct Communication Scenario}\label{pexdir}
In the first time slot of the DC scenario, all $S_2$ source nodes decode successfully a direct message from their nearest $S_1$ neighbor with a probability denoted as $p_{DC_1}$. Similarly, with $p_{DC_2}$ we denote the probability that all $S_1$ source nodes decode successfully a direct message from their nearest $S_2$ neighbor in the second time slot. These probabilities (i.e., $p_{DC_1}$ and $p_{DC_2}$) are independent and have common network parameters except for the intensity $\lambda_1$ and $\lambda_2$, respectively. Therefore, the probability $p_{DC_{\varphi}}=f(\lambda_{\varphi})$ is a function of the intensity and the probability $p_{DC}$ that all source nodes have successfully decoded a message from the nearest neighbor of the other type is given by
\begin{equation}
p_{DC}=p_{DC_1}p_{DC_2}=\prod_{\varphi=1}^2p_{DC_\varphi}
\end{equation}
To that end, to derive $p_{DC}$ we have to calculate the probability $p_{DC_\varphi}$. Moreover, in order to account for the power splitting process described by \eqref{powersplitfunction}, we have to differentiate between the cases of $h<\psi$ and $h\geq\psi$. Therefore, the probability of successful message exchange for the direct scenario is given by the following theorem.

\begin{theorem}\label{firsttheorem}
The probability of successful message decoding in one time slot for the DC scenario is given by \eqref{theorem1},
\begin{figure*}[!t]
\normalsize
\begin{equation}
\label{theorem1}
\begin{split}
&p_{DC_\varphi}=\pi\lambda_{\varphi}(1-e^{-\mu\psi})\int_0^{\infty}\exp{\bigg[-\pi\lambda_{\varphi} r\bigg(1+{\gamma^*}^{2/\alpha}\int_{{\gamma^*}^{-2/\alpha}}^{\infty}\frac{1}{1+u^{a/2}}\text{d}u\bigg)-\frac{\mu \gamma^*N}{\psi P_t}r^{\alpha/2}\bigg]}\text{d}r+\\
&+2\pi\lambda_{\varphi} e^{-\mu\psi}\int_0^{\infty}\exp{\bigg[-2\pi\lambda_{\varphi}\frac{{}_2F_1\big(1,\frac{\alpha-2}{\alpha};\frac{2\alpha-2}{\alpha};\frac{\gamma^*P_t\psi}{P_t\psi-r^{\alpha}\gamma^* N}\big)\big(\frac{1}{\gamma^*r^{\alpha}}-\frac{N}{P_t\psi}\big)^{-\frac{2}{\alpha}} \big(\frac{\gamma^*P_t\psi}{P_t\psi-r^{\alpha}\gamma^* N}\big)^{\frac{\alpha-2}{\alpha}} }{\alpha-2}-\pi\lambda_{\varphi} r^2\bigg]}r\text{d}r
\end{split}
\end{equation}
\hrulefill
\end{figure*}
where ${}_2F_1(a,b;c;z)=\sum_{n=0}^{\infty}\frac{(a)_n(b)_n}{(c)_n}\frac{z^n}{n!}$ is the hypergeometric function.
\end{theorem}
\begin{proof}
By taking into account \eqref{powersplitfunction} and \eqref{sinr}, the probability $p_{DC_{\varphi}}$ is given by
\begin{equation}
p_{DC_{\varphi}}=\text{Pr}(\gamma>\gamma^*\cap h<\psi)+\text{Pr}(\gamma>\gamma^*\cap h\geq\psi).
\label{pex_dintersection}
\end{equation}
Conditioning on the value of the RV $h$ using the Kolmogorov definition of conditional probabilities, we obtain 
\begingroup
\footnotesize
\begin{equation}
p_{DC_{\varphi}}=\text{Pr}(h<\psi)\text{Pr}(\gamma>\gamma^*|h<\psi)+\text{Pr}(h\geq\psi)\text{Pr}(\gamma>\gamma^*|h\geq \psi).
\label{pex_d}
\end{equation}\endgroup
Since $h$ is exponentially distributed with rate $\mu$, \eqref{pex_d} can be written as
 \begin{equation}
p_{DC_{\varphi}}\hspace{-0.5ex}=\hspace{-0.5ex}\bigg(\hspace{-0.5ex}1-\frac{1}{e^{\mu\psi}}\hspace{-0.5ex}\bigg)\text{Pr}(\gamma>\gamma^*|h<\psi)+\frac{1}{e^{\mu\psi}}\text{Pr}(\gamma>\gamma^*|h\geq \psi).
\label{pex_d2}
\end{equation}
In \eqref{pex_d2}, the probability $\text{Pr}(\gamma>\gamma^*|h<\psi)$ can be easily calculated using guidelines from \cite{coverage} and it is given as
\begingroup
\footnotesize
\begin{align}
&\text{Pr}(\gamma>\gamma^*|h<\psi)=\\
&=\pi\lambda_{\varphi}\int_0^{\infty}\exp{\bigg[-\pi \lambda_{\varphi}  r\bigg(1+\int_{{\gamma^*}^{\frac{-2}{\alpha}}}^{\infty}\frac{{\gamma^*}^{2/\alpha}}{1+u^{a/2}}\text{d}u\bigg)-\frac{\mu \gamma^*N}{P_t r^{-\alpha/2}}\bigg]}\text{d}r,\nonumber
\end{align}
\endgroup
whereas the proof for the probability $\text{Pr}(\gamma>\gamma^*|h\geq \psi)$ is provided in Appendix A. Replacing $\text{Pr}(\gamma>\gamma^*|h<\psi)$ and $\text{Pr}(\gamma>\gamma^*|h\geq \psi)$ in \eqref{pex_d2}, concludes the proof.
\end{proof}
\begin{lemma}\label{lemmaa4}
For the special but common case when the path loss exponent is $\alpha=4$, Theorem \ref{firsttheorem} is simplified into 
\begingroup
\footnotesize
\begin{align}
&p_{DC_{\varphi}}\hspace{-0.5ex}=\pi\lambda_{\varphi}(1-e^{-\mu\psi})\sqrt{\frac{\pi}{\omega(\gamma^*)}}\exp{\bigg(\frac{\chi(\lambda_{\varphi},\gamma^*)^2}{4\omega(\gamma^*)}\bigg)}Q\bigg(\frac{\chi(\lambda_{\varphi},\gamma^*)}{\sqrt{2\omega(\gamma^*)}}\bigg)+\nonumber\\
&+2\pi\lambda_{\varphi} e^{-\mu\psi}\int_0^{\infty}\exp{\bigg[-\pi\lambda_{\varphi} r^2\Big(1+\zeta(r,\gamma^*)\arctan\big[\zeta(r,\gamma^*)\big] \Big)\bigg]}r\text{d}r,
\label{theorem1le}
\end{align}
\endgroup
where $Q(x)=\frac{1}{\sqrt{2\pi}}\int_x^\infty\exp{(-q^2/2)}\text{d}q$ is the tail probability of the standard normal distribution, $\chi(\lambda_{\varphi},\gamma^*)=\pi\lambda_{\varphi}(1+\sqrt{\gamma^*}(\pi/2-\arccot{(\sqrt{\gamma^*})}))$, $\zeta(r,\gamma^*)=\sqrt{\frac{P_t\gamma^*}{P_t-r^4\gamma^*N}}$ and $\omega(\gamma^*)=\mu\gamma^*N/P_t$.
\end{lemma}
\begin{proof}
The proof of Lemma \ref{lemmaa4} is provided in Appendix B.
\end{proof}

\subsection{Cooperative Communication Scenario}

In the case of the cooperative scenario, the two sets of sources exchange their messages either directly or with the assistance of relay nodes. Therefore, the overall probability of successful exchange in the cooperative case, denoted as $p_{CC}$, depends both on the probabilities $p_{DC_1}$ and $p_{DC_2}$ derived in Section~\ref{pexdir} and on the probability $p_{CC_{R_{\varphi}}}$, which denotes the probability that relay has decoded a message from its nearest type $\varphi$ source and a type $\hat{\varphi}$ source node has successfully decoded this message through this relay. Hence, there are three events for successful exchange in the cooperative scenario: i) both directly and through a relay, ii) only directly, or iii) only through a relay. Since these events are mutually exclusive, the probability of successful exchange in the cooperative case is given by the following lemma\footnote{It should be noted that, although the interference at the relay and destination in the two first time slots comes from the same set of nodes, the impact of fading minimizes the correlation and, therefore, the events can be considered independent.}.

\begin{lemma}\label{prop12}
The probability of successful message exchange in one CP for the cooperative scenario is given by
\begin{equation}
p_{CC}=\big(p_{DC_1}+p_{CC_{R_1}}(1-p_{DC_1})\big)\big(p_{DC_2}+p_{CC_{R_2}}(1-p_{DC_2})\big).
\label{psdprop}
\end{equation}
\end{lemma}
\begin{proof}
The proof of Lemma \ref{prop12} is provided in Appendix C.
\end{proof}
\begin{remark}\label{remdccc}
In interference-limited systems, thermal noise is not an important consideration that results in a weak dependence of the probability of successful transmission $p_{DC_{\varphi}}$ with the node intensity \cite{coverage}. To that end, it follows that $p_{DC_1}\simeq p_{DC_2}$ and, thus, $p_{CC}\simeq\big(p_{DC_{\varphi}}+p_{DC_{\varphi}}^2-p_{DC_{\varphi}}^3\big)$. From the latter, it can be easily proven that $p_{CC}\geq p_{DC}$ holds always. Still, although it is always more probable to achieve a successful message exchange in the CC scenario, this result does not imply higher performance of the CC scenario in the end-to-end performance of the network. Consequently, in the following, we perform an analysis on the network lifetime and other end-to-end performance metrics (e.g., spatial throughput) to identify trade-offs between the two scenarios.
\end{remark}
\section{Network Lifetime}\label{section_lifetime}

One of the most important metrics for a WSN is its operating lifetime. In this section, the analysis for the derivation of the network lifetime and the average harvested power is given for all scenarios. In this way, it becomes possible to determine the gains of WEH using DPS.

\subsection{Direct Communication Scenario}\label{dirscen}
After $w_d\in\mathbb{N}_0$ communication periods and without taking EH into account, the average battery level of a source node in the DC scenario is defined by the amount of energy $E_{con}$ consumed per CP and it is given by
\begin{equation}
\bar{B_d}(w_d)=B_I-w_dE_{con}=B_I-w_dt_s(P_r+P_t),
\label{levd}\end{equation}
where $B_I$ is the initial energy level, $t_s$ is the duration of a time slot, $P_r$ is the power consumption at the reception mode, and $P_t$ is the power consumption at the transmission mode. 
In the case that the source nodes have EH capabilities, their battery level is increased in each CP by the average harvested power per CP denoted as $\bar{P_d}^{EH}$. Thus, 
\begin{equation}
\bar{B_d}^{EH}(w_d^{EH})=B_I-w_d^{EH}t_s(P_r+P_t)+w_d^{EH}t_s\bar{P_d}^{EH}.
\label{levehd}
\end{equation}
The roots of  \eqref{levd} and \eqref{levehd} (i.e., the values of $w_d$ and $w_d^{EH}$ that the battery is discharged) provide the source node's average lifetime in CPs $\bar{L_d}$ and $\bar{L_d}^{EH}$, respectively
\begin{equation}
\bar{L_d}=\frac{B_I}{t_s(P_r+P_t)}
\label{wmax}
\end{equation}and
\begin{equation}
\bar{L_d}^{EH}=\frac{B_I}{[t_s(P_r+P_t-\bar{P_d}^{EH})]_+},
\label{wmaxeh}
\end{equation}
where $[\xi]_+=\text{max}(\xi,0)$. 
\begin{remark}
In the extreme case that the denominator of \eqref{wmaxeh} is equal to zero, the consumed power is lower or equal than the average harvested power and, hence, the network lifetime becomes infinite (i.e., the perpetual network operation).
\end{remark}
In the following theorem, the average harvested power $\bar{P_d}^{EH}$ of a source node is provided, in order to complete the derivation of the average network lifetime with EH in the DC scenario $\bar{L_d}^{EH}$, given in \eqref{wmaxeh}. 

\begin{theorem}\label{ahptheorem}
The average harvested power in one CP of a type $S_\varphi$ source node at the DC scenario, while taking into account DPS and before the RF-to-DC conversion efficiency is described by
\begingroup
\footnotesize
\begin{equation}
\bar{P}_{DPS_{d\varphi}}=P_te^{-\mu\psi}\Bigg[\frac{\pi\alpha\lambda_{\hat{\varphi}}}{\mu(\alpha-2)}+\psi e^{-\mu\psi}\text{Ei}[-\mu\psi]\bigg[\frac{\pi\alpha\lambda_{\hat{\varphi}}}{\alpha-2}-\mathbb{E}\big\{r_{c\hat{\varphi}}^{-\alpha}\big\} \bigg]\Bigg],\nonumber\label{avpowharvbef}
\end{equation}\endgroup
whereas the actual average harvested power after applying the RF-to-DC conversion efficiency is given by
\begin{equation}
\bar{P_{d\varphi}}^{EH}=\bar{P}_{DPS_{d\varphi}}\Big[a_3\big(\bar{P}_{\log}\big)^3+a_2\big(\bar{P}_{\log}\big)^2+a_1\big(\bar{P}_{\log}\big)+a_0\Big],\label{avpowharv}
\end{equation} 
where $\bar{P}_{\log}=10\log_{10}{\frac{\bar{P}_{DPS_{d\varphi}}}{1\text{mW}}}$, $\text{Ei}[x]=-\int_{-x}^{\infty}\frac{e^{-t}}{t}\text{d}t$ for nonzero values of $x$ denotes the exponential integral and $\mathbb{E}\big\{r_{c\hat{\varphi}}^{-\alpha}\big\}$ denotes the expected value of the path loss to the nearest type $S_{\hat{\varphi}}$ transmitter for different path loss exponent values $\alpha>2$, given within the proof.
\end{theorem}
\begin{proof}The proof of Theorem \ref{ahptheorem} is provided in Appendix D.
\end{proof}
\begin{remark}\label{minrem}At this point, it should be mentioned that the average lifetime $\bar{L_d}^{EH}$ is limited by the set of sources with the least average harvested power. Thus, it holds that\begin{equation}\bar{L_d}^{EH}={B_I\big/[t_s(P_r+P_t-\min\{\bar{P_{d1}}^{EH}, \bar{P_{d2}}^{EH}\})]_+}.\end{equation}
This happens because when a set of sources consumes all of its energy, then we assume that the system has reached its lifetime.
\end{remark}

\subsection{Cooperative Communication Scenario}
In the cooperative communication scenario, a set of relay nodes assists the source nodes to exchange their messages. Therefore, without taking EH into account, the battery level of a node after $w_c\in\mathbb{N}_0$ CPs in the cooperative scenario is defined by the initial battery level and the amount of energy $E_{con}$ consumed per CP and it is given by
\begin{equation}
\bar{B_c}(w_c)=B_I-w_cE_{con}=B_I-w_ct_s\Big(2P_r+P_t\big(1+\mathds{1}_{R}\big)\Big),
\label{levcomb}\end{equation}
where $\mathds{1}_{R}$ is the indicator function that determines whether \eqref{levcomb} represents the battery level of a relay node or a source and it is described by
  \begin{numcases}{\mathds{1}_{R}=}\label{bbbb}
    1, &  \text{Relay node.}\\
    0, &  \text{Source node.}\nonumber
  \end{numcases}
Similarly, in the case that the nodes have EH capabilities, their battery level at any CP $w_c^{EH}$ is
\begin{equation}
\bar{B_c}^{EH}(w_c^{EH})=B_I-w_c^{EH}t_s\Big(2P_r+P_t\big(1+\mathds{1}_{R}\big)\Big)+w_c^{EH}t_s\bar{P_c}^{EH},\label{levehcomb}
\end{equation}
where $\bar{P_c}^{EH}$ is the average harvested power in one CP.
The roots of \eqref{levcomb} and \eqref{levehcomb} provide the node's average lifetime in CPs for each case, respectively
\begin{equation}
\bar{L_c}=\frac{B_I}{t_s\Big(2P_r+P_t\big(1+\mathds{1}_{R}\big)\Big)}
\label{wmaxcomb}
\end{equation}and
\begin{equation}
\bar{L_c}^{EH}=\frac{B_I}{[t_s\Big(2P_r+P_t\big(1+\mathds{1}_{R}\big)\Big)-t_s\bar{P_c}^{EH}]_+}.
\label{wmaxehcomb}
\end{equation}
As in the DC scenario, the average harvested power $\bar{P_c}^{EH}$ must be derived to complete the calculation of the network lifetime with EH in the CC scenario $\bar{L_c}^{EH}$, given in \eqref{wmaxehcomb}.

\begin{lemma}\label{ahplemma}
The average harvested power of a type $S_{\varphi}$ source $\bar{P}_{DPS_{c\varphi}}$ or a relay node $\bar{P}_{DPS_{cR}}$ for the cooperative scenario, while taking into account DPS and before the RF-to-DC conversion efficiency, is the sum of the average power harvested by the transmissions of the other two sets. Hence, we obtain \eqref{avpowharvbefc} and \eqref{avpowharvbefc2},
\begin{figure*}[!t]
\normalsize
\begin{equation}
\label{avpowharvbefc}
\begin{split}
\bar{P}_{DPS_{c\varphi}}=P_te^{-\mu\psi}\Bigg[\frac{\pi\alpha(\lambda_R+\lambda_{\hat{\varphi}})}{\mu(\alpha-2)}+\psi e^{-\mu\psi}\text{Ei}[-\mu\psi]\bigg[\frac{\pi\alpha(\lambda_R+\lambda_{\hat{\varphi}})}{\alpha-2}-\mathbb{E}\big\{r_{c\hat{\varphi}}^{-\alpha}\big\}-\mathbb{E}\big\{r_{cR}^{-\alpha}\big\} \bigg]\Bigg]
\end{split}
\end{equation}
\end{figure*}
\begin{figure*}[!t]
\normalsize
\begin{equation}
\label{avpowharvbefc2}
\begin{split}
\bar{P}_{DPS_{cR}}=P_te^{-\mu\psi}\Bigg[\frac{\pi\alpha(\lambda_\varphi+\lambda_{\hat{\varphi}})}{\mu(\alpha-2)}+\psi e^{-\mu\psi}\text{Ei}[-\mu\psi]\bigg[\frac{\pi\alpha(\lambda_\varphi+\lambda_{\hat{\varphi}})}{\alpha-2}-\mathbb{E}\big\{r_{c\hat{\varphi}}^{-\alpha}\big\}-\mathbb{E}\big\{r_{c\varphi}^{-\alpha}\big\} \bigg]\Bigg]
\end{split}
\end{equation}
\hrulefill
\end{figure*}
where $\text{Ei}[x]$ is the exponential integral of $x$ and $\mathbb{E}\big\{r_{cR}^{-\alpha}\big\}$ denotes the expected value of the path loss to the nearest relay. 

The actual average harvested power after applying the RF-to-DC conversion efficiency $\bar{P_{c\iota}}^{EH}$, $\iota\in\{\varphi,R\}$ is given by
\begin{equation}
\bar{P_{c\iota}}^{EH}=\bar{P}_{DPS_{c\iota}}\Big[a_3\big(\bar{P}_{c\log}\big)^3+a_2\big(\bar{P}_{c\log}\big)^2+a_1\big(\bar{P}_{c\log}\big)+a_0\Big],\label{avpowharvc}
\end{equation} 
where $\bar{P}_{c\log}=10\log_{10}{\frac{\bar{P}_{DPS_{c\iota}}}{1\text{mW}}}$.
\end{lemma}
\begin{proof}The same line of thought is followed for this proof as in Theorem \ref{ahptheorem}. However, for the cooperative case, the sources are assisted by a set of relays. Therefore, each source node receives on average energy from two sets (i.e., in one timeslot from the relay transmissions and in another timeslot from the transmissions of the other set of sources). Moreover, the relays are receiving the energy from the transmissions of the two source sets. Thus, the average harvested power while taking into account DPS and before the RF-to-DC conversion efficiency of an $S_\varphi$ source is 
\begin{equation}
\bar{P}_{DPS_{c\varphi}}=\bar{P}_{DPS_{d{\varphi}}}+\bar{P}_{DPS_{dR}},\label{avpowharvbefcinproof1}
\end{equation}
where $\bar{P}_{DPS_{dR}}$ can be derived from $\bar{P}_{DPS_{d{\varphi}}}$ using $\lambda_R$ as the intensity. For a relay node the average harvested power is
\begin{equation}
\bar{P}_{DPS_{cR}}=\bar{P}_{DPS_{d_1}}+\bar{P}_{DPS_{d_2}}.\label{avpowharvbefcinproofR}
\end{equation}
Substituting \eqref{avpowharvbefcinproof1} or \eqref{avpowharvbefcinproofR} to \eqref{avpowharvc} and following a procedure as in Theorem \ref{ahptheorem}, yields the respective actual average harvested power after applying the RF-to-DC conversion efficiency, which concludes the proof.
\end{proof}
Thus, by combining \eqref{avpowharvc} with \eqref{wmaxehcomb}, the maximum lifetime of a node with EH in the cooperative scenario can be derived. Similar to Remark \ref{minrem}, the average lifetime in the CC scenario is defined by the minimum between $\bar{P_{c1}}^{EH}$ and $\bar{P_{d2}}^{EH}$.

\section{Optimal Intensity and Performance metrics}\label{section_metrics}
In this section, we will introduce the optimal intensity, which provides an accurate estimation of the number of nodes per unit area needed to achieve the highest possible lifetime for the network, and two metrics that are useful for evaluating the performance of the network, i.e., the spatial throughput that indicates the average number of messages exchanged per unit area and the total messages exchanged on average.
 \subsection{Optimal Intensity}
 In previous works with WEH networks that do not take into account the RF-to-DC conversion efficiency, the network intensity is a monotonic function of the average harvested power. However, in a more realistic approach where the antennas are not ideal, as the network intensity and, thus, the interference 
 increases, the average harvested power rises to a local maximum and then decreases due to the low RF-to-DC conversion efficiency. Therefore, it is important to know the network topology characteristics such as the intensity of the transmitting set of nodes that achieves the maximum average harvested power for the receiving set of nodes. The optimization problem considered can be described as
 \begin{equation}
\begin{aligned}
& \underset{\lambda_{\hat{\varphi}}}{\text{max}}
& & \bar{P_{d\varphi}}^{EH}(\lambda_{\hat{\varphi}}) \\
& \text{s.t.} & &  \lambda_{\hat{\varphi}}\geq0 \\
& & &  0\leq\epsilon(P_I)\leq1 \\
\end{aligned}
\end{equation}
 and a solution of this problem is given in the following Lemma.
 
 \begin{lemma}\label{loptlemma}
 The optimal intensity $\lambda_{opt}$ to achieve maximum lifetime in a network with DPS and RF-to-DC conversion efficiency described by \eqref{powersplitfunction} and \eqref{conveff}, respectively, is given by
 \begingroup
\footnotesize
 \begin{equation}
\lambda_{opt}\hspace{-0.5ex}=\hspace{-0.5ex}\frac{\mu(\alpha-2)10^{-\frac{\alpha_2}{30\alpha_3}}}{10^3P_t\pi\alpha e^{-\mu\psi+1}}\exp\hspace{-0.7ex}{\bigg[\frac{2^{2/3}f}{60\alpha_3}+\frac{\ln{^210}(\alpha_2^2-3\alpha_1\alpha_3)+900\alpha_3^2}{2^{-4/3}60\alpha_3f}\bigg]},\nonumber
\end{equation}\endgroup
where
\begin{equation}
f=\sqrt[3]{-\rho+\sqrt{\rho^2-4\big(\ln^210(\alpha_2^2-3\alpha_1\alpha_3)+900\alpha_3^2\big)^3}}\nonumber
\end{equation}
and
\begin{equation}
\rho=\ln^310\big(27\alpha_0\alpha_3^2-9\alpha_1\alpha_2\alpha_3+2\alpha_2^3\big)+54000\alpha_3^3\nonumber
\end{equation}
 \end{lemma}
 \begin{proof}
 The proof of Lemma \ref{loptlemma} is provided in Appendix E.
 \end{proof}
 \begin{remark}
 It should be noted that the optimal intensity of the $S_1$ source nodes calculated using Lemma \ref{loptlemma} maximizes the lifetime of the $S_2$ set of nodes. Similarly, the optimal intensity of the $S_2$ set of nodes maximizes the lifetime of the $S_1$ set.
 \end{remark}
 \subsection{ST and TME}
The probability of successful exchange derived in Section \ref{section_pex} for all scenarios is a throughput metric for the link under examination. In order to have a complete picture of the network performance, we employ the metric of spatial throughput \cite[5.3.1]{book}, \cite{spatial}, which provides an average of the throughput over all the links in the network. Hence, the spatial throughput (messages/s/unit-area) of the network is defined as
\begin{equation}
S_{sc}=\frac{(\lambda_1+\lambda_2)p_{sc}}{g_{sc}t_s}\text{ (messages/s/unit-area)}\label{spatthr},
\end{equation}
where $sc=\{DC, CC\}$, $p_{sc}$ and $g_{sc}$ denote the successful exchange probability and the number of slots per scenario, respectively.

Finally, another metric that can be deduced using the spatial throughput is the average total messages exchanged in a lifetime per unit area (TME), which is given by multiplying the spatial throughput with the network lifetime and the number of slots per CP for each scenario. TME can be written as
\begin{equation}
\text{TME}_{sc}=S_{sc}w_{sc}g_{sc}\text{ (messages/unit-area)}\label{tme},
\end{equation}
where $w_{sc}$ denotes the network lifetime for the various scenarios derived in Section \ref{section_lifetime}. In the following section, we will present and validate the numerical results of all the metrics that have been presented so far.
\section{Analytical and Simulation Results}\label{section_results}
In this section, we validate the proposed theoretical framework via extensive simulations and provide useful insights on the use of WEH by comparing the metrics of interest for the different communication scenarios.

\subsection{Simulation Setup}
We compare the two proposed scenarios, direct and simple cooperative without EH (DC and CC, respectively) and with EH (DC-EH and CC-EH, respectively). For high accuracy, we create 10.000 realizations of a $500$ m by $500$ m area with intensities varying from $0.01$ to $0.5$ devices per m$^2$ (i.e., the number of devices per realization is from $3.000$ up to $150.000$). The time slot duration is denoted as $t_s$ and depends on the application scenario and the chosen bitrate. The transmit power is $P_t=75$ mW, while the power for the reception mode is $P_r=100$ mW \cite{ptpr} and the initial level of a node's battery is $L_I=1000$~J. Additionally, the path loss exponent is chosen to be $\alpha=4$, although it is possible to use any value $\alpha>2$. For the model validation, the channel fading gain is set to $\mu=1$ and the noise power to $N=-124$~dBm for $100$~kHz system bandwidth for all scenarios, unless otherwise stated, while we vary the values of decoding threshold $\gamma^*$ and intensity $\lambda$ in order to present the performance of the system under different conditions.  In addition, if not explicitly stated otherwise, the decoding threshold is fixed at $\gamma^*=0$ dB and the intensity $\lambda$ of the PPPs is set to $\lambda_1=0.1$, $\lambda_2=0.5$ and $\lambda_R=0.25$.
\begin{figure}[!t]
\centering
\includegraphics[width=1\columnwidth]{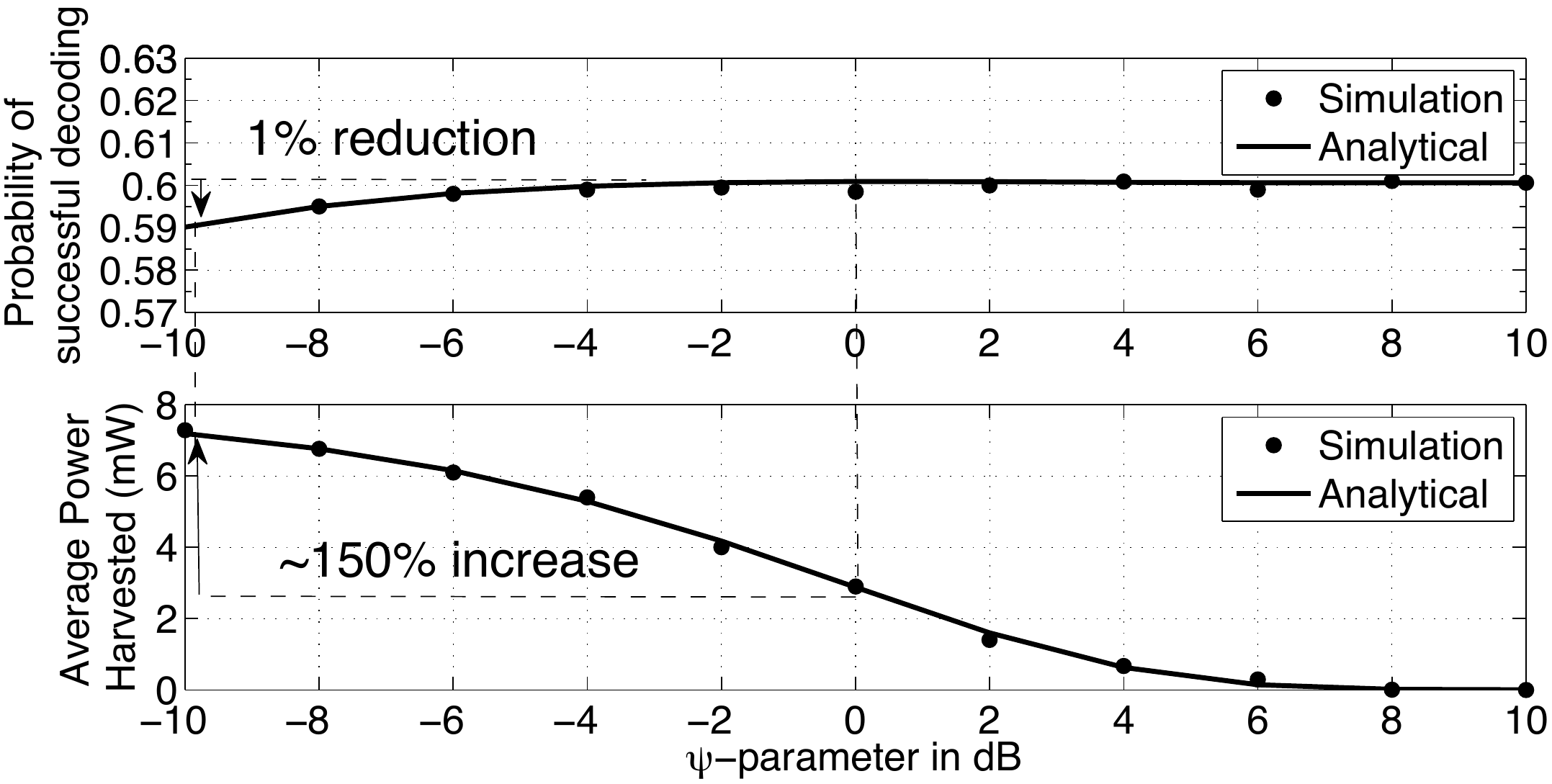}
\caption{Comparison of probability $p_{DC}$ and average harvested power $\bar{P}_d^{EH}$ versus the $\psi$-parameter.}
\label{final11}
\end{figure}

Moreover, in all the experiments, the coefficients for the RF-to-DC conversion efficiency $\epsilon$ given in \eqref{conveff} are $\alpha_3=-4.6\cdot10^{-5}$, $\alpha_2=-7.8\cdot10^{-4}$, $\alpha_1=0.03$ and $\alpha_0=0.62$, according to \cite{rectenna2} for the case of $940$ MHz. Regarding the $\psi$-parameter in \eqref{powersplitfunction}, since it defines the amount of power that is split between the harvester and the information receiver, it can be chosen in a way that increases the average harvested power without affecting the probability of successful exchange. In Fig.~\ref{final11}, we provide the relation of $\psi$ with the two metrics (i.e., probability of successful decoding and average harvested power). It can be observed that by sacrificing only $1\%$ in the probability of successful decoding, the average harvested power is increased by $\sim150\%$. This is due to the fact that the probability of exchange drops with a low rate as $\psi$ is reduced, while the average harvested power rises with a much higher rate. Therefore, in our experiments, the $\psi$-parameter has been fixed at $-10$ dB or $\psi=0.1$. 
\subsection{Model Validation and Performance Evaluation}

In this section, we validate the basic metrics (i.e., probability of successful message exchange and average harvested power) of our analysis, that are used for the derivations of the end-to-end QoS and lifetime metrics. In Fig.~\ref{final12}, we plot the probability of successful message exchange for the direct and cooperative communication scenarios versus the decoding threshold $\gamma^*$.
\begin{figure}[!t]
\centering
\includegraphics[width=1\columnwidth]{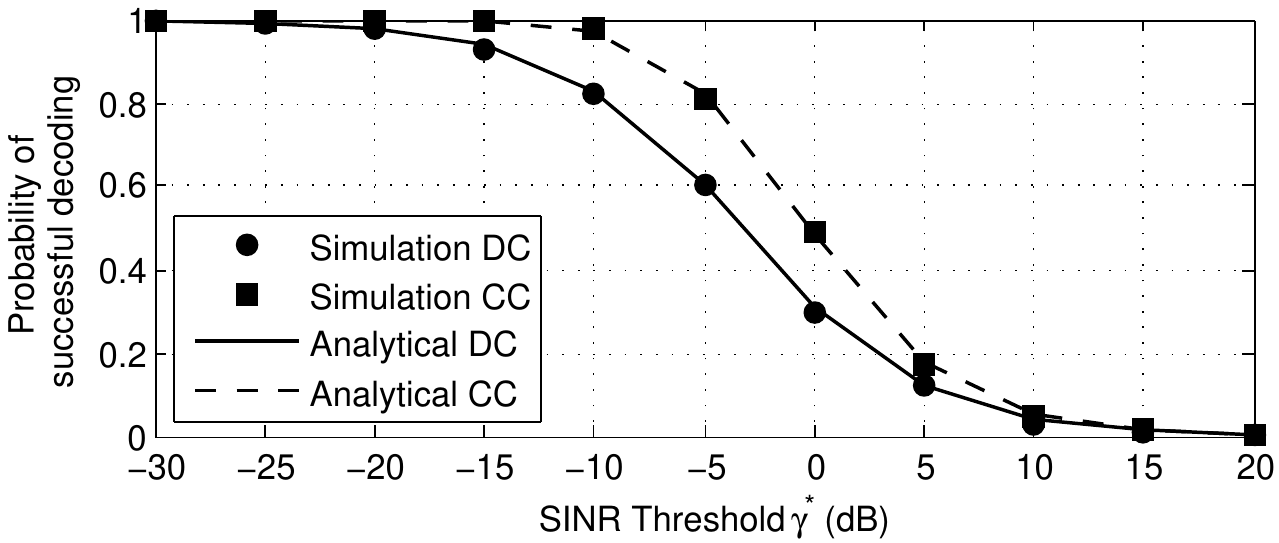}
\caption{Probability of successful message exchange vs. decoding threshold $\gamma^*$ for the direct and cooperative scenarios.}
\label{final12}
\end{figure}
As we can see, the probability $p_{DC}$ matches perfectly with the simulations and, thus, Theorem~1 is validated. Furthermore, the probability $p_{DC}$ becomes lower as the decoding threshold increases. This result can be justified by the fact that, for higher decoding thresholds, the received signal must be much stronger than the interference plus noise. Similar conclusions can be derived in the result for the cooperative communication scenario. As we can see, Lemma~2 is validated and the probability $p_{CC}$ decreases for higher decoding thresholds. By comparing the two curves, we can also notice that the probability of successful exchange is higher in the cooperative communication case compared to the direct one for the same decoding thresholds. This has been already proven in small-scale networks and with our study we extend this result even for large-scale networks with random relay deployment. Thus, thanks to diversity, there is a probability that the message exchange will take place via relay nodes, even if the direct communication fails.

\begin{figure}[!t]
\centering
\includegraphics[width=1\columnwidth]{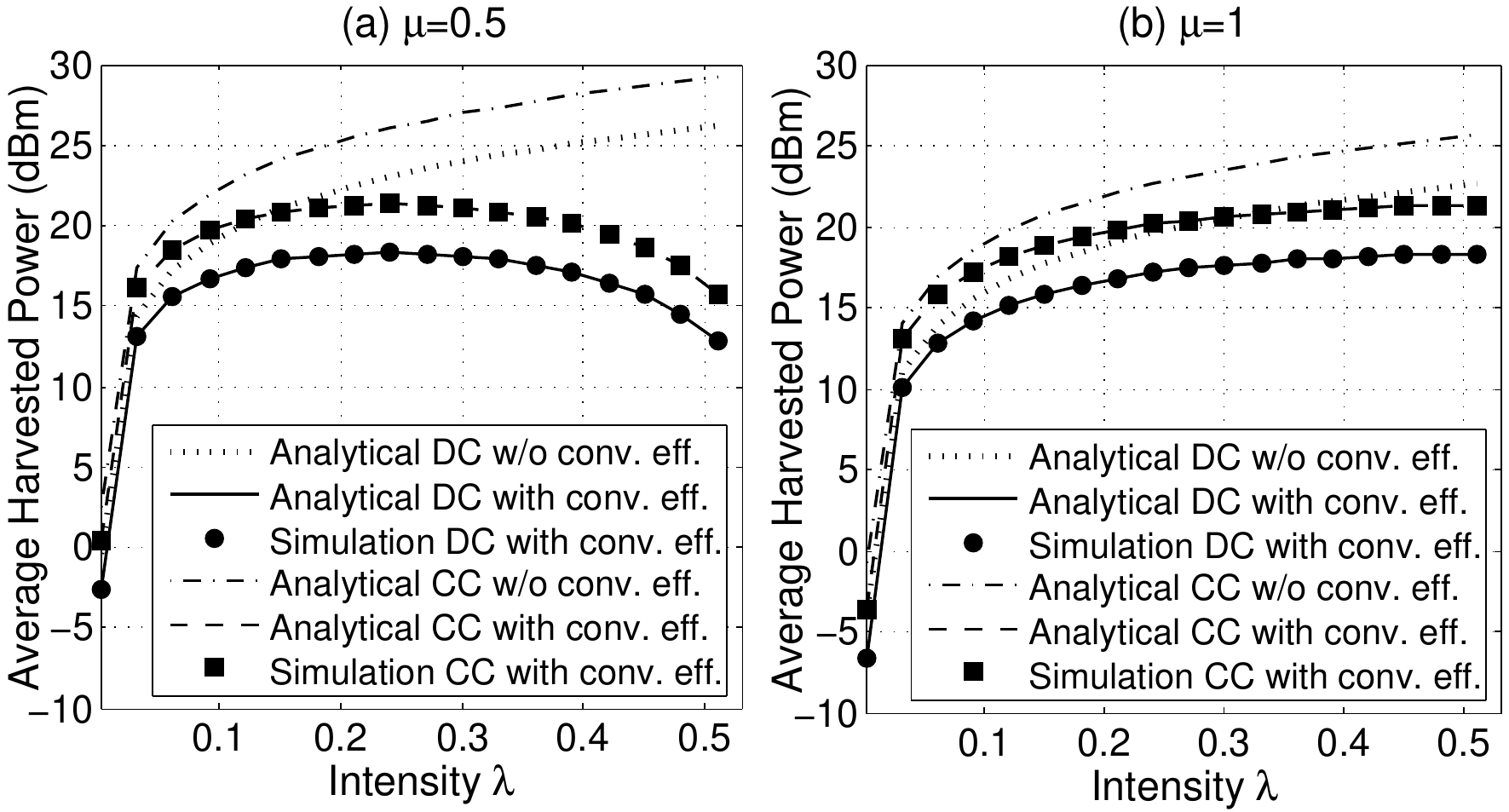}
\caption{Average harvested power vs. Intensity. (a) $\mu=0.5$, (b) $\mu=1$.}
\label{final145}
\end{figure}
In Fig.~\ref{final145}, we plot the average harvested power by a source in one CP versus the node intensity, considering two different cases for the channel conditions, a) favorable with $\mu=0.5$ and b)~moderate with $\mu=1$. One first straightforward observation from both figures is that, as the intensity increases until a certain point, the nodes harvest more power, due to the higher interference. Also, compared to Fig. \ref{final145}a, the results in Fig. \ref{final145}b need higher intensity to achieve the same average harvested power, because the fading conditions attenuate the received power and, thus, the average harvested power. However, it is very interesting to see that, after a peak value, the average harvested power is decreasing. This can be seen clearly in Fig.~\ref{final145}a and it stems from the fact that the RF-to-DC conversion efficiency of the rectennas, given in \eqref{conveff} and shown in Fig.~\ref{rfdceff}, decreases as the received power increases over a certain point. Indeed, to highlight the difference between the average harvested power with and without RF-to-DC conversion, we also plot in the same figure the cases without the conversion, which show the significant amount of energy that is lost due to the conversion (e.g., for $\mu=0.5$ and $\lambda=0.2$ in the DC scenario, the difference between the two cases is over 3 dBm). This is a very important insight which implies that i) adding more nodes in the network does not necessarily increase the lifetime of the network and ii) there is a unique maximum of the average harvested power according to the conditions of the system. In addition, by comparing the different communication scenarios in both figures, we notice that the cooperative scenario provides the highest amount of harvested power. This is due to the fact that, in this scenario, there are also relays that provide more energy to the system in one CP.
\begin{figure}[!t]
\centering
\includegraphics[width=1\columnwidth]{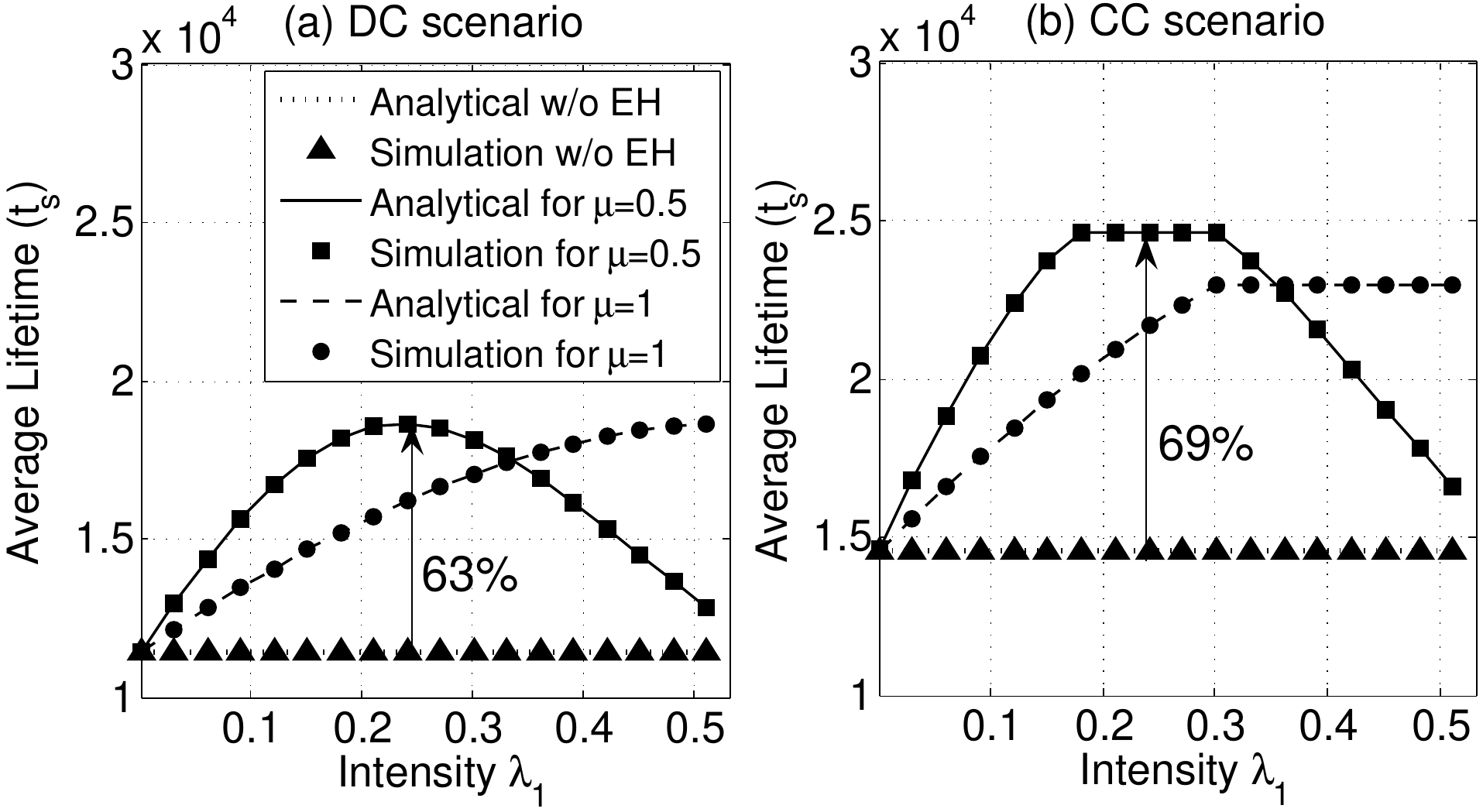}
\caption{Average Lifetime vs. Intensity: (a) Comparison between DC and DC-EH, (b) Comparison between CC and CC-EH.}
\label{al}
\end{figure}

In Fig.~\ref{al}, we present the average network lifetime with and without EH for both scenarios versus the intensity $\lambda_1$ of the $S_1$ source nodes. For the DC scenario (Fig. \ref{al}a), we assume that the intensity $\lambda_2$ is equal to the optimal intensity calculated using Lemma \ref{loptlemma} (i.e., $\lambda_2\simeq0.25$ for $\mu=0.5$ and $\lambda_2\simeq0.5$ for $\mu=1$). Similarly, for the CC scenario (Fig. \ref{al}b), we assume that the intensity of the relays is equal to the optimal ($\lambda_R=0.25$) and we set $\lambda_2=0.3$. As expected, EH increases the lifetime of the network, especially for the cooperative scenario, where the lifetime gains can reach up to $69$\%, compared to a gain of $63\%$ in Fig. \ref{al}a. The gains are higher for the CC scenario, because relays contribute to the average harvested energy during each CP compared to the DC scenario. Additionally, it can be noticed that, in the CC case, there is a limit in the average lifetime from the intensities between $0.2$ and $0.3$. This stems from the fact that $S_2$ sources cannot achieve higher lifetime than this limit (i.e., $\lambda_2=0.3$), which limits the lifetime of the whole network, as it is explained in Remark \ref{minrem}.

\begin{figure}[!t]
\centering
\includegraphics[width=1\columnwidth]{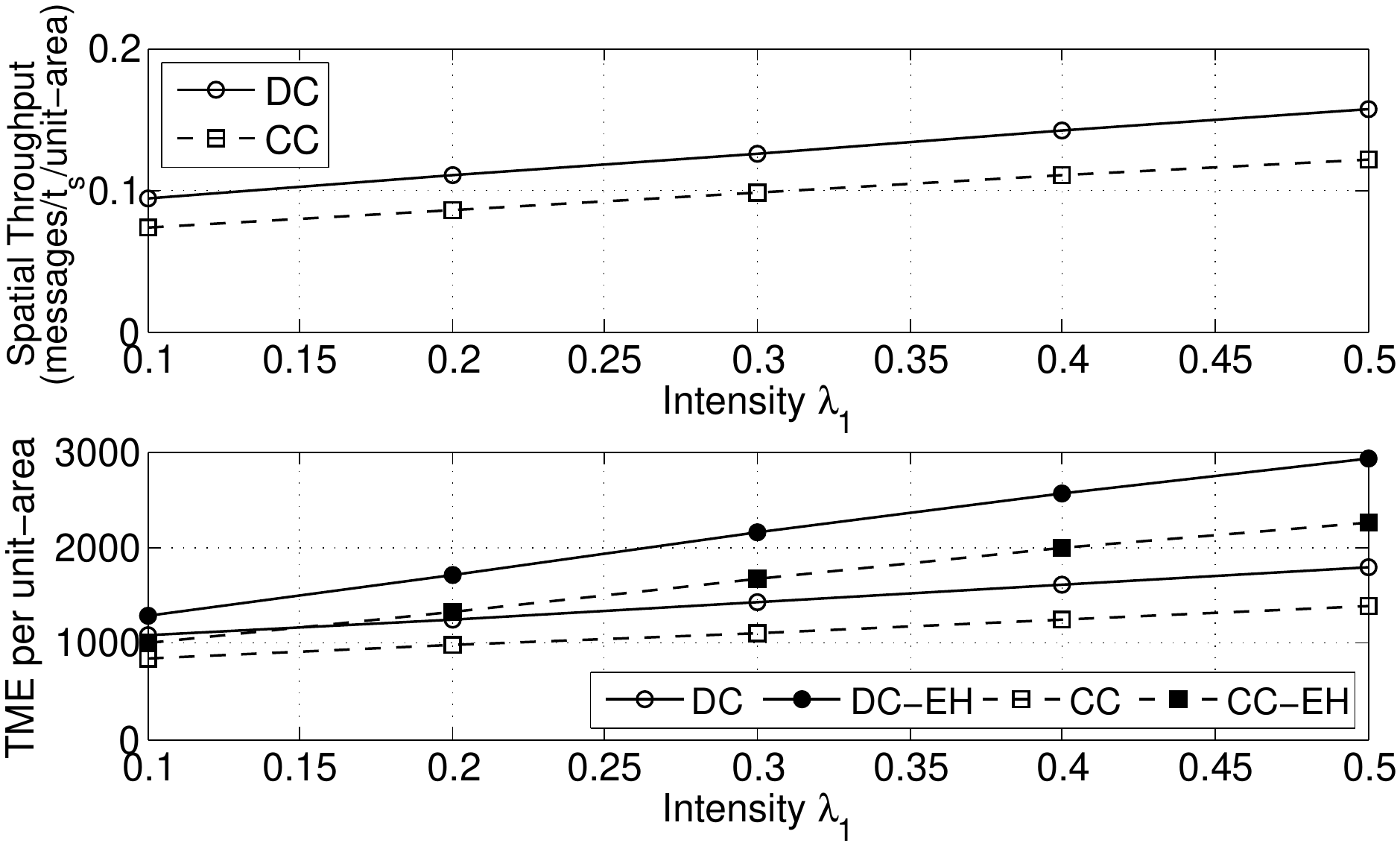}
\caption{(a) Spatial throughput vs. Intensity and (b) Successfully exchanged messages in a lifetime vs. Intensity for the different communication scenarios.}
\label{tmef}
\end{figure}
Having validated the analysis, we now present a performance evaluation for the two communication scenarios in Fig.~\ref{tmef}. In this figure, the simulation results appear as markers while the lines represent the analytical results. As depicted in Fig.~\ref{tmef}a, the spatial throughput increases with the intensity, since more nodes exchange messages per unit area. Moreover, it is interesting to notice that, although the probability of message exchange is always higher in the cooperative communication (see also Remark \ref{remdccc}), the spatial throughput for the cooperative scenario presents lower performance than the DC scenario. This can be justified by considering the randomness in the deployment of the relays and the longer CPs in the cooperative scenario. To clarify, although the performance gains from cooperation are obvious in a scenario where the relays are located in between the source nodes, this is not the case for randomly deployed networks. In such networks, it is possible for a direct link to provide better communication than a cooperative link, whereas the performance of the cooperative scenario is limited and depends on the random relay deployment. This fact in conjunction with the longer CPs in the cooperative scenario are the reasons that the message exchange rate of CC drops in comparison to the DC scenario.

Moreover, in Fig.~\ref{tmef}b, we combine the two metrics given in Fig.~\ref{al} and Fig.~\ref{tmef}a and estimate the number of successfully decoded messages during the network lifetime per unit area as a function of the intensity. From this figure, it is evident that the CC-EH scenario presents lower performance compared to the direct scenario with EH, which shows that the additional time slots in the CC scenario drop the performance. However, in Fig.~\ref{tmef}b, it is worth noting that the performance of the network through time is not taken into account. Since the battery capacity of the CC scenario is decreased through time with a lower rate than in the DC scenario, we could identify the trade-offs between the two scenarios while taking into account the total exchanged messages and the average lifetime.

\begin{figure}[!t]
\centering
\includegraphics[width=1\columnwidth]{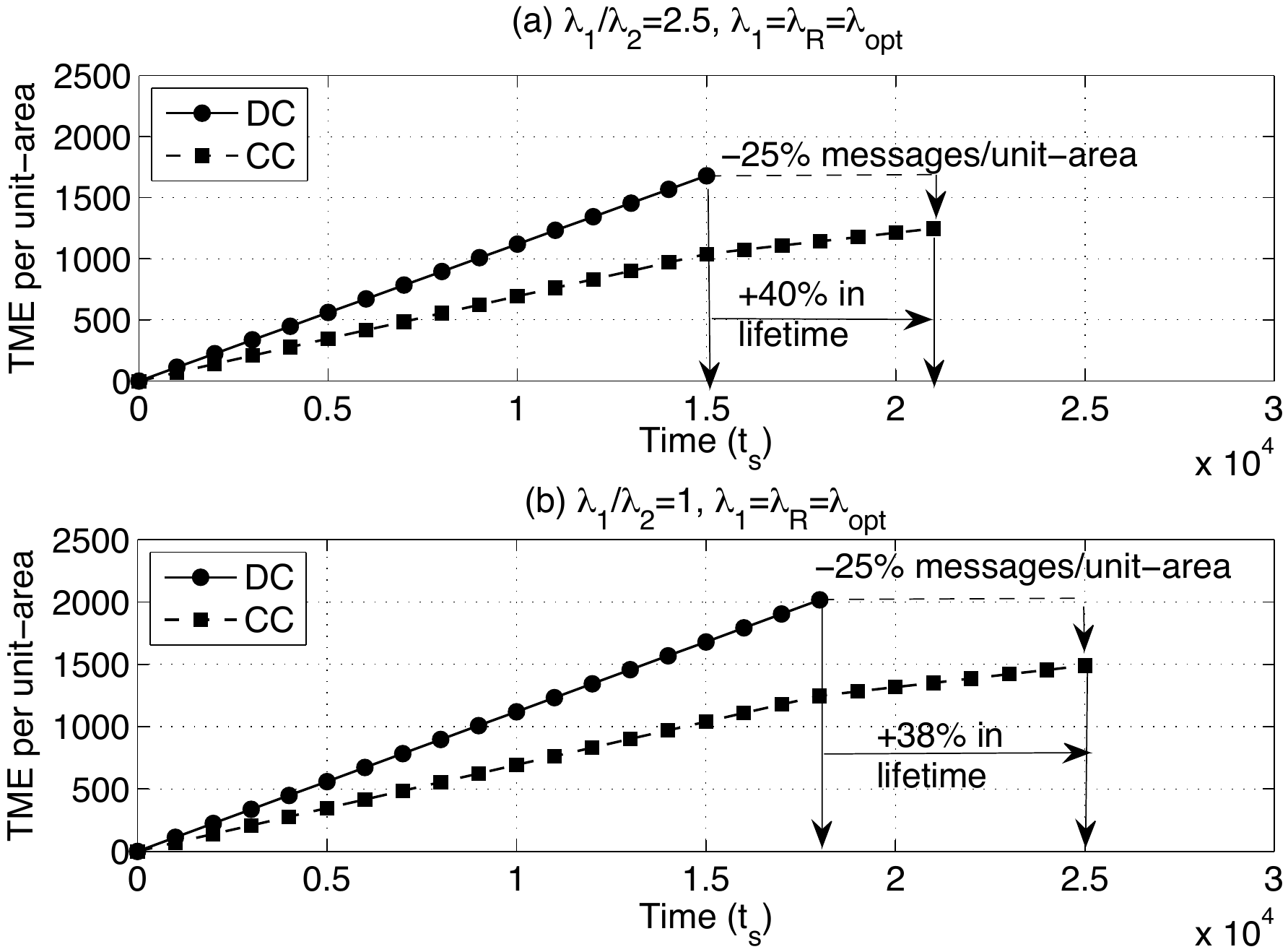}
\caption{Successfully exchanged messages per unit-area vs. Time for the different scenarios.}
\label{mpt}
\end{figure}

Finally, in Fig.~\ref{mpt}, we present the average exchanged messages per unit area versus time for two different intensity combinations (i.e., in Fig.~\ref{mpt}a, $\lambda_1=\lambda_R=\lambda_{opt}$ and $\lambda_2=\lambda_{opt}/2.5$ and in Fig.~\ref{mpt}b, $\lambda_1=\lambda_2=\lambda_R=\lambda_{opt}$). We observe that, in Fig. \ref{mpt}a, the network has lower lifetime compared to Fig. \ref{mpt}b, because the network lifetime is limited by the lower intensity of the $S_2$ set of source nodes. On the other hand, when all sets have the optimal intensity (Fig. \ref{mpt}b), the network lifetime is maximized. Moreover, it can be clearly seen that the communication scenarios present different trade-offs. For instance, in Fig. \ref{mpt}a, the DC scenario has higher number of exchanged messages but lower lifetime, while the CC scenario demonstrates higher lifetime ($+40$\%) with fewer exchanged messages ($-25$\%). Similarly, in Fig. \ref{mpt}b, the CC scenario demonstrates higher lifetime ($+38$\%) with again fewer exchanged messages ($-25$\%). 

To that end, the results in Fig.~\ref{tmef} and Fig.~\ref{mpt} reveal the counter-intuitive insight that the DC scenario presents better communication performance than the CC scenario in randomly deployed dense networks. Nevertheless, thanks to its higher lifetime, the CC scenario could be proved ideal for applications such as in cases where the nodes are embedded in buildings or bodies without easy access, where longevity is more important than high data rates. 
\section{Conclusion}\label{section_conclusion}
This paper has studied the impact of WEH on the information exchange in large-scale networks. The purpose of the randomly deployed WSN nodes is to exchange successfully their messages locally with their neighbors, either directly (direct communication scenario) or through a relay node (cooperative communication scenario). The different scenarios were compared in terms of message exchange probability, spatial throughput and network lifetime. The theoretical derivations were validated by extensive Monte-Carlo simulations. Finally, the comparison of the two scenarios highlighted the importance of WEH in large-scale networks and revealed that the direct communication scenario presents better communication performance than the cooperative scenario in randomly deployed dense networks. However, the cooperative scenario is more advisable in applications where longevity matters, since it is superior in terms of lifetime.


%
\appendices
\section{Proof of $\text{Pr}(\gamma>\gamma^*|h\geq\psi)$ in Theorem 1}
In this section, we will derive the probability $\text{Pr}(\gamma>\gamma^*|h\geq\psi)$. Conditioning on the nearest transmitting source at a distance $r$, the probability of successful message reception given that $h\geq\psi$ is given by
\begin{align}
&\text{Pr}(\gamma>\gamma^*|h\geq\psi)=\mathbb{E}_r[\text{Pr}(\gamma>\gamma^*|h\geq\psi,r)]=\nonumber\\
&=\int_0^{\infty}\text{Pr}(\gamma>\gamma^*|h\geq\psi,r)f_r(r)\text{d}r=\nonumber\\
&=\int_0^{\infty}\text{Pr}\bigg(\frac{P_thr^{-\alpha}v(\psi)}{v(\psi)I_r+N}>\gamma^*\Big|r\bigg)f_r(r)\text{d}r=\nonumber\\
&=\int_0^{\infty}\text{Pr}\bigg(h>\frac{\gamma^*r^{\alpha}I_r}{1-\phi\gamma^*r^{\alpha}}>\gamma^*\Big|r\bigg)f_r(r)\text{d}r,\nonumber
\end{align}
where $f_r(r)$ denotes the probability density function (PDF) of $r$, given in \cite[2.9.1]{book} and $\phi=N/(P_t\psi)$. Since $h$ follows an exponential distribution, we have
\begingroup
\footnotesize
\begin{align}
&\text{Pr}(\gamma>\gamma^*|h\geq\psi)=\int_0^{\infty}\mathbb{E}_{I_r}\bigg[\text{Pr}\bigg(h>\frac{\gamma^*r^{\alpha}I_r}{1-\phi\gamma^*r^{\alpha}}>\gamma^*\Big|r\bigg)f_r(r)\text{d}r=\nonumber\\
&=\int_0^{\infty}\mathbb{E}_{I_r}\bigg[\exp\bigg(-\frac{\mu\gamma^*r^{\alpha}}{1-\phi\gamma^*r^{\alpha}}I_r\bigg)\bigg|r\bigg]f_r(r)\text{d}r=\nonumber\\
&=\int_0^{\infty}\mathcal{L}_{I_r}\bigg(\frac{\mu\gamma^*r^{\alpha}}{1-\phi\gamma^*r^{\alpha}}\bigg)f_r(r)\text{d}r,\label{prbeflapl}
\end{align}\endgroup
where $\mathcal{L}_{I_r}(s)$ defines the Laplace transform of the interference. We aim to calculate the Laplace transform by applying the following steps:
\begingroup
\footnotesize
\begin{align}
&\mathcal{L}_{I_r}\bigg(\frac{\mu\gamma^*r^{\alpha}}{1-\phi\gamma^*r^{\alpha}}\bigg)=\mathbb{E}e^{-\big(\frac{\mu\gamma^*r^{\alpha}}{1-\phi\gamma^*r^{\alpha}}\big)I_r}=\mathbb{E}\bigg[\prod_{i\in\Phi/x}e^{-\big(\frac{\mu\gamma^*r^{\alpha}}{1-\phi\gamma^*r^{\alpha}}\big)\frac{h}{r^{\alpha}}}\bigg],\nonumber
\end{align}\endgroup
where $x$ denotes the transmitting source which is excluded from the aggregated interference. Since the fading is iid,
\begin{align}
&\mathcal{L}_{I_r}\bigg(\frac{\mu\gamma^*r^{\alpha}}{1-\phi\gamma^*r^{\alpha}}\bigg)=\mathbb{E}_{\Phi}\Bigg[\prod_{i\in\Phi/x}\mathbb{E}_h\Big[e^{-\big(\frac{\mu\gamma^*r^{\alpha}}{1-\phi\gamma^*r^{\alpha}}\big)hr^{-\alpha}}\Big]\Bigg]=\nonumber\\
&=\mathbb{E}_{\Phi}\Bigg[\prod_{i\in\Phi/x}\frac{\mu}{\mu+\big(\frac{\mu\gamma^*r^{\alpha}}{1-\phi\gamma^*r^{\alpha}}\big)r^{-\alpha}}\Bigg].\nonumber
\end{align}
Using the probability generating functional for the PPP $\Phi$ \cite[4.6]{book}, we obtain
\begingroup
\footnotesize
\begin{align}
&\mathcal{L}_{I_r}\bigg(\frac{\mu\gamma^*r^{\alpha}}{1-\phi\gamma^*r^{\alpha}}\bigg)=\exp{\bigg(-\int_r^{\infty}\frac{2\pi\lambda_{\varphi}}{1+u^{\alpha}(1/(\gamma^*r^{\alpha})-\phi)}u\text{d}u\bigg)}=\label{lapla1}\\
&=\exp{\Bigg(-2\pi\lambda_{\varphi}\frac{\big(\frac{1}{\gamma^*r^{\alpha}}-\phi\big)^{-2/\alpha}\text{ }_2F_1\big(1,\frac{\alpha-2}{\alpha};2-\frac{2}{\alpha};\frac{\gamma^*}{r^{\alpha}\gamma^*\phi-1}\big)}{(\alpha-2)\big(\frac{\gamma^*}{1-r^{\alpha}\gamma^*\phi}\big)^{\frac{2-\alpha}{\alpha}}}\Bigg)},\label{lapla}
\end{align}\endgroup
where the integral in \eqref{lapla1} is derived with the aid of a computational software program\footnote{Wolfram Research, Inc., \emph{Mathematica,} Version 10.0, Champaign, IL, 2014.} and the hypergeometric function $_2F_1(a,b;c;z)$ is valid for $|z|<1$ which holds for realistic WSN scenarios.
Combining \eqref{lapla} with \eqref{prbeflapl}, yields the result to the probability $\text{Pr}(\gamma>\gamma^*|h\geq\psi)$.

\section{Proof of Lemma 1}
For the special case that the path loss exponent is $\alpha=4$, Theorem \ref{firsttheorem} can be further simplified. Considering \eqref{pex_d2}, the probability $\text{Pr}(\gamma>\gamma^*|h<\psi)$ for $\alpha=4$ is derived in \cite{coverage}. Regarding $\text{Pr}(\gamma>\gamma^*|h\geq\psi)$, we will simplify it by  using Euler's transformation formula for the hypergeometric function $_2F_1$ \cite[15]{hypergeom}:
\begin{equation}\label{euler}
_2F_1(a,b;c;z)=(1-z)^{c-a-b}\cdot\text{}_2F_1(c-a,c-b;c;z).\nonumber
\end{equation}
Therefore, in our case 
\begin{equation}
_2F_1\bigg(1,\frac{1}{2};\frac{3}{2};\frac{\gamma^*P_t\psi}{P_t\psi-r^{4}\gamma^* N}\bigg)\hspace{-0.7ex}=\hspace{-0.5ex}_2F_1\bigg(\frac{1}{2},1;\frac{3}{2};\frac{\gamma^*P_t\psi}{P_t\psi-r^{4}\gamma^* N}\bigg)\nonumber
\end{equation}
Moreover, by applying the hypergeometric representation of $\arctan$
\begin{equation}
\frac{\arctan(z)}{z}=_2F_1\bigg(\frac{1}{2},1;\frac{3}{2};-z^2\bigg),\nonumber
\end{equation}
we obtain the result of Lemma \ref{lemmaa4}.

\section{Proof of Lemma 2}
In order to derive the probability $p_{CC}$ in \eqref{psdprop}, we need to calculate the probabilities $p_{CC_{R_{\varphi}}}$, $\varphi\in\{1,2\}$, that the relay has successfully received from a message from a type $\varphi$ source and delivered it to a node of the other type. This means that the probability $p_{CC_{R_{\varphi}}}$ is the product of two independent probabilities of successful direct communication transmissions, i) from an $S_{\varphi}$ source to the relay, denoted as $p_{S_{\varphi}\rightarrow R}$, and ii) from the relay to an $S_{\hat{\varphi}}$ source which we will denote as $p_{R\rightarrow S_{\hat{\varphi}}}$. As we have stated in the system model, the probability of successful transmission is defined as the probability that the SINR $\gamma$ measured at the nearest receiver is higher than a threshold $\gamma^*$. Since each of the single transmission of the CC scenario is described by the same principles as in the DC scenario, the probability $p_{DC_{R}}$ is derived following the same line of though and employing the same mathematical tools as in Section \ref{pexdir}. Hence, the probability $p_{S_{\varphi}\rightarrow R}$ is equal to the probability $p_{DC_{\varphi}}$, which is the probability that any node will decode successfully a message from its nearest $S_{\varphi}$ source. Moreover, $p_{R\rightarrow S_{\hat{\varphi}}}$ is also a direct transmission probability derived in the same way as $p_{DC_{\varphi}}$ using the intensity $\lambda_R$ instead of $\lambda_{\varphi}$.

Therefore, the probability of successful message delivery through a relay for the CC scenario is given by
\begin{equation}
p_{CC_{R_{\varphi}}}=p_{S_{\varphi}\rightarrow R}\cdot p_{R\rightarrow S_{\hat{\varphi}}}=p_{DC_R}\cdot p_{DC_{\varphi}}.\label{psr_eq_pexd2}
\end{equation}
Combining \eqref{psr_eq_pexd2} and \eqref{theorem1} with \eqref{psdprop}, we obtain the probability of successful message exchange in the cooperative communication scenario.

\section{Proof of Theorem 2}
According to the power splitter rule provided in \eqref{powersplitfunction}, the harvester receives $(1-v(\psi)\big)100\%$ of the total aggregated received power and only when $h_{c\hat{\varphi}}>\psi$, where $h_{c\hat{\varphi}}$ denotes the channel fading gain of the nearest transmitting node. For simplicity, we drop the $\varphi$ notation and for the rest $h_c=h_{c\hat{\varphi}}$, $\lambda=\lambda_{\hat{\varphi}}$, $r_c=r_{c\hat{\varphi}}$ and $\bar{P}_{DPS_{d}}=\bar{P}_{DPS_{d\varphi}}$. Therefore, the average harvested power of a source node at the DC scenario, while considering DPS, is provided by 
\begin{equation}\label{beflin}
\bar{P}_{DPS_d}=\text{Pr}(h_c>\psi)\cdot\mathbb{E}\Bigg\{\bigg(1-\frac{\psi}{h_c}\bigg)\sum_{i\in\Phi}P_th_ir_i^{-\alpha}\bigg|h_c>\psi\Bigg\}.
\end{equation}
Using the linearity property of the expected value on \eqref{beflin} and considering that $h_c$ follows an exponential distribution with mean value $1/\mu$, we get
\begingroup
\footnotesize
\begin{equation}\label{aflin}
\bar{P}_{DPS_d}=\frac{P_t}{e^{\mu\psi}}\cdot\Bigg(\mathbb{E}\bigg\{\underbrace{\sum_{i\in\Phi}\frac{h_i}{r_i^{\alpha}}\Big|h_c>\psi}_\text{A}\bigg\}-\psi\mathbb{E}\bigg\{\underbrace{\sum_{i\in\Phi}\frac{h_i}{h_cr_i^{\alpha}}\bigg|h_c>\psi}_\text{B}\bigg\}\Bigg).
\end{equation}\endgroup
To derive the expected values of \eqref{aflin}, we could employ Campbell's theorem on sums \cite[4.2]{book}. However, the expected values in \eqref{aflin} are conditioned on $h_c$, which means that the channel fading channel of the nearest transmitter has to be higher than a certain $\psi$ value, i.e., $h_c>\psi$ for $i=c$. Hence, in order to be able to apply Campbell's theorem, we will employ the following procedure. By expanding the sum $A$ in the expected value, we obtain
\begin{align}
&\mathbb{E}\bigg\{\underbrace{\sum_{i\in\Phi}h_ir_i^{-\alpha}\Big|h_c>\psi}_\text{A}\bigg\}=\nonumber\\
&=\mathbb{E}\Big\{h_1r_1^{-\alpha}+\dots+h_cr_c^{-\alpha}+\dots+h_ir_i^{-\alpha}\Big|h_c>\psi\Big\}=\nonumber\\
&=\mathbb{E}\big\{h_1r_1^{-\alpha}\big\}+\dots+\mathbb{E}\big\{h_cr_c^{-\alpha}|h_c>\psi\big\}+\dots+\mathbb{E}\big\{h_ir_i^{-\alpha}\big\}.\label{onlyc}
\end{align}
As it can be seen in \eqref{onlyc}, only the received power of the nearest transmitter is affected by this condition. Thus, by adding and subtracting an equivalent average received power without the channel fading gain condition for this term, we obtain
\begin{align}
&\mathbb{E}\bigg\{\underbrace{\sum_{i\in\Phi}\frac{h_i}{r_i^{\alpha}}\Big|h_c>\psi}_\text{A}\bigg\}=\mathbb{E}\bigg\{\frac{h_1}{r_1^{\alpha}}\bigg\}+\dots+\mathbb{E}\bigg\{\frac{h_c}{r_c^{\alpha}}|h_c>\psi\bigg\}+\nonumber\\
&+\Big(\mathbb{E}\big\{h_cr_c^{-\alpha}\big\}-\mathbb{E}\big\{h_cr_c^{-\alpha}\big\}\Big)+\dots+\mathbb{E}\big\{h_ir_i^{-\alpha}\big\}=\nonumber\\
&=\mathbb{E}\big\{h_1r_1^{-\alpha}\big\}+\dots+\mathbb{E}\big\{h_cr_c^{-\alpha}\big\}+\dots+\mathbb{E}\big\{h_ir_i^{-\alpha}\big\}+\nonumber\\
&+\mathbb{E}\big\{h_c|h_c>\psi\big\}\mathbb{E}\big\{r_c^{-\alpha}\big\}-\mathbb{E}\big\{h_c\big\}\mathbb{E}\big\{r_c^{-\alpha}\big\}=\nonumber\\
&=\mathbb{E}\bigg\{\sum_{i\in\Phi}h_ir_i^{-\alpha}\bigg\}+\mathbb{E}\big\{r_c^{-\alpha}\big\}\cdot\Big(\mathbb{E}\big\{h_c|h_c>\psi\big\}-\mathbb{E}\big\{h_c\big\}\Big).
\label{onlycfix}
\end{align}
In \eqref{onlycfix} the expectation of the sum can be easily derived using Campbell's theorem. Moreover, the conditional expectation of the exponentially distributed RV in \eqref{onlycfix} is given by
\begin{align}
\mathbb{E}\big\{h_c|h_c>\psi\big\}&=\frac{\int_0^{\infty}h_c\mu e^{-\mu h_c}\mathds{1}(h_c>\psi)\text{d}h_c}{\int_0^{\infty}\mu e^{-\mu h_c}\mathds{1}(h_c>\psi)\text{d}h_c}=\nonumber\\
&=\frac{\int_{\psi}^{\infty}h_c e^{-\mu h_c}\text{d}h_c}{\int_{\psi}^{\infty} e^{-\mu h_c}\text{d}h_c}=\frac{1+\psi\mu}{\mu},\nonumber
\end{align}
 where $\mathds{1}(h_c>\psi)$ is the indicator function. Thus, by applying Campbell's theorem on sums and due to the independence between the RVs $h_i$ and $r_i$, \eqref{onlycfix} yields
\begin{align}\label{aftercamp}
&\mathbb{E}\bigg\{\underbrace{\sum_{i\in\Phi}h_ir_i^{-\alpha}\Big|h_c>\psi}_\text{A}\bigg\}=\frac{1}{\mu}\mathbb{E}\bigg\{\sum_{i\in\Phi}r_i^{-\alpha}\bigg\}+\psi\mathbb{E}\big\{r_c^{-\alpha}\big\}=\nonumber\\
&=\frac{\lambda}{\mu}\int_{\mathbb{R}_d}r^{-\alpha}\text{d}r+\psi\mathbb{E}\big\{r_c^{-\alpha}\big\}=\frac{\pi\alpha\lambda}{\mu(\alpha-2)}+\psi\mathbb{E}\big\{r_c^{-\alpha}\big\}.
\end{align}
Following a similar procedure as in \eqref{onlycfix}, sum $B$ in \eqref{aflin} is given by
\begingroup
\footnotesize
\begin{align}
&\mathbb{E}\bigg\{\hspace{-0.5ex}\underbrace{\sum_{i\in\Phi}\frac{h_i}{h_cr_i^{\alpha}}\bigg|h_c>\psi}_\text{B}\hspace{-1ex}\bigg\}\hspace{-0.5ex}=\mathbb{E}\bigg\{\hspace{-0.5ex}\frac{h_1}{h_cr_1^{\alpha}}\bigg|h_c>\psi\hspace{-0.5ex}\bigg\}\hspace{-0.5ex}+\dots+\mathbb{E}\bigg\{\hspace{-0.5ex}\frac{h_c}{h_cr_c^{\alpha}}\bigg|h_c>\psi\hspace{-0.5ex}\bigg\}+\nonumber\\
&+\dots+\mathbb{E}\bigg\{\frac{h_i}{h_c}r_i^{-\alpha}\bigg|h_c>\psi\bigg\}=\mathbb{E}\bigg\{h_1r_1^{-\alpha}\bigg\}\mathbb{E}\bigg\{\frac{1}{h_c}\bigg|h_c>\psi\bigg\}+\nonumber\\
&+\dots+\mathbb{E}\big\{r_c^{-\alpha}\big\}+\bigg(\mathbb{E}\bigg\{h_cr_c^{-\alpha}\bigg\}\mathbb{E}\bigg\{\frac{1}{h_c}\bigg|h_c>\psi\bigg\}-\nonumber\\
&-\mathbb{E}\bigg\{\frac{h_c}{r_c^{\alpha}}\bigg\}\mathbb{E}\bigg\{\frac{1}{h_c}\bigg|h_c>\psi\bigg\}\bigg)+\dots+\mathbb{E}\bigg\{h_ir_i^{-\alpha}\bigg\}\mathbb{E}\bigg\{\frac{1}{h_c}\bigg|h_c>\psi\bigg\}=\nonumber\\
&=\mathbb{E}\bigg\{\sum_{i\in\Phi}h_ir_i^{-\alpha}\bigg\}\mathbb{E}\bigg\{\frac{1}{h_c}\bigg|h_c>\psi\bigg\}-\mathbb{E}\big\{h_cr_c^{-\alpha}\big\}\mathbb{E}\bigg\{\frac{1}{h_c}\bigg|h_c>\psi\bigg\}+\nonumber\\
&+\mathbb{E}\big\{r_c^{-\alpha}\big\}=\mathbb{E}\bigg\{\frac{1}{h_c}\bigg|h_c>\psi\bigg\}\cdot\bigg(\frac{\pi\alpha\lambda}{\mu(\alpha-2)}-\frac{1}{\mu}\mathbb{E}\big\{r_c^{-\alpha}\big\}\bigg)+\mathbb{E}\big\{r_c^{-\alpha}\big\}\label{sumB}
\end{align}\endgroup
Once again, the conditional probability in \eqref{sumB} is given by
\begingroup
\footnotesize
\begin{equation}\label{conexp2}
\mathbb{E}\bigg\{\frac{1}{h_c}\bigg|h_c>\psi\bigg\}=\frac{\int_0^{\infty}\frac{1}{h_c}\mu e^{-\mu h_c}\mathds{1}(h_c>\psi)\text{d}h_c}{\int_0^{\infty}\mu e^{-\mu h_c}\mathds{1}(h_c>\psi)\text{d}h_c}=-\mu e^{\mu\psi}Ei[-\mu\psi],
\end{equation}\endgroup
where $\text{Ei}[x]=-\int_{-x}^{\infty}\frac{e^{-t}}{t}\text{d}t$ for nonzero values of $x$ denotes the exponential integral. Moreover, the expected value of the path loss to the nearest transmitter $\mathbb{E}\{r_{c}^{-\alpha}\}$ is provided by
\begingroup
\footnotesize
\begin{equation}\label{rcalphaexp}
\mathbb{E}\{r_c^{-\alpha}\}=\int_0^{\infty}r_c^{-\alpha}f_{r_c}(r_c)\text{d}r_c=1-e^{-\lambda\pi}+\int_1^{\infty}r_c^{1-\alpha}2\pi\lambda e^{-\pi\lambda r_c^2}\text{d}r_c,
\end{equation}\endgroup
where $f_{r_c}(r_c)$ denotes the PDF of the distance to the nearest neighbor, given in \cite[2.9.1]{book}. The integral in \eqref{rcalphaexp} can be solved for any value of $\alpha>2$, e.g.:
\begingroup
\footnotesize
\begin{align}\label{alphas}
&\alpha=3: &\mathbb{E}\{r_c^{-3}\}=&1-e^{-\lambda\pi}+2\lambda\pi\bigg (e^{-\pi\lambda}-\pi\sqrt{\lambda}\cdot\text{Erfc}\Big[\sqrt{\lambda\pi}\Big]\bigg)\nonumber\\
&\alpha=4: &\mathbb{E}\{r_c^{-4}\}=&1-e^{-\lambda\pi}+\lambda\pi\Big(e^{-\lambda\pi}+\lambda\pi\cdot\text{Ei}[-\lambda\pi]\Big)\nonumber\\
&\alpha=5: &\mathbb{E}\{r_c^{-5}\}=&1-e^{-\lambda\pi}+\frac{2}{3}\lambda\pi\bigg(\frac{1-2\lambda\pi}{e^{\lambda\pi}}+\frac{2\pi^2}{\lambda^{-\frac{3}{2}}}\cdot\text{Erfc}\Big[\sqrt{\lambda\pi}\Big]\bigg)\nonumber
\end{align}\endgroup
where $\text{Erfc}[x]=(\frac{2}{\sqrt{\pi}})\int_x^{\infty}e^{-t^2}\text{d}t$ denotes the complementary error function.

Combining \eqref{aftercamp}, \eqref{sumB} and \eqref{conexp2} into \eqref{aflin}, the proof is concluded.
\section{Proof of Lemma 4}
Due to i) the fact that $\bar{P}_{DPS_{d}}$ is monotonically increasing with the intensity and ii) the concave nature of \eqref{conveff}, we know that there is one local maximum for $\lambda_{\hat{\varphi}}>0$ and $0\leq\epsilon(P_I)\leq1$. Therefore, by taking the derivative of $\bar{P}_{{d{\varphi}}}^{EH}$ in Theorem \ref{ahptheorem} with respect to $\lambda_{\hat{\varphi}}$ and solving the equation
\begin{equation}
\frac{\vartheta\bar{P}_{{d{\varphi}}}^{EH}(\lambda_{\hat{\varphi}})}{\vartheta{\lambda_{\hat{\varphi}}}}=0,
\end{equation}
we obtain the value of $\lambda_{\hat{\varphi}}>0$ for which the lifetime of the network is maximized.

\ifCLASSOPTIONcaptionsoff
  \newpage
\fi

\end{document}